%% file: arXiv-teamltl.tex
\documentclass[acmsmall,screen,nonacm,authorversion=true]{acmart}
\AtBeginDocument{%
  }

\input{macros}
\input{teamltl-packages}
\newtheorem{remark}{Remark}

\begin{document}
\input{content.tex}

\end{document}

%% file: macros.tex
\newcommand{\problemFont}[1]{\mathsf{#1}}
\newcommand{\classFont}[1]{\mathsf{#1}}

\newcommand{\LTL}{\protect\ensuremath{\mathrm{LTL}}\xspace}
\newcommand{\CTL}{\protect\ensuremath{\mathrm{CTL}}\xspace}
\newcommand{\TeamLTL}{\protect\ensuremath{\mathrm{TeamLTL}}\xspace}
\newcommand{\TeamCTL}{\protect\ensuremath{\mathrm{TeamCTL}}\xspace}

\newcommand{\dfn}{\coloneqq}
\newcommand{\ddfn}{\Coloneqq}

\newcommand{\ltlopFont}[1]{\mathsf{#1}}
\newcommand{\E}{\ltlopFont{E}}
\newcommand{\A}{\ltlopFont{A}}
\newcommand{\X}{\ltlopFont{X}}
\newcommand{\F}{\ltlopFont{F}}
\newcommand{\G}{\ltlopFont{G}}
\newcommand{\U}{\ltlopFont{U}}
\newcommand{\R}{\ltlopFont{R}}

\newcommand{\nhyltlmodels}{\not\models}
\newcommand{\hyltlmodels}{\models}
\newcommand{\ltlmodels}{\models}
\newcommand{\ctlmodels}{\models}
\newcommand{\ctlmodelsnew}{\models^c}

\newcommand{\N}{\mathbb N}
\newcommand{\ap}{\mathrm{AP}}

\newcommand{\TPC}{\mathrm{PC}}
\newcommand{\TMC}{\mathrm{MC}}
\newcommand{\TSAT}{\mathrm{SAT}}
\newcommand{\TVAL}{\mathrm{VAL}}
\newcommand{\LTLMC}{\mathrm{MC(LTL)}}
\newcommand{\LTLPC}{\mathrm{PC(LTL)}}

\newcommand{\leqpm}{\leq_{\mathrm{m}}^{\mathrm{p}}}
 
\newcommand{\PSPACE}{\protect\ensuremath{\mathrm{PSPACE}}} 

\newcommand{\ATIME}{\protect\ensuremath{\mathrm{ATIME\text-ALT}}}

\newcommand{\pol}{\protect\ensuremath{\mathrm{pol}}}
\newcommand{\EXPTIME}{\protect\ensuremath{\mathrm{EXPTIME}}}
\newcommand{\NEXPTIME}{\protect\ensuremath{\mathrm{NEXPTIME}}}
\newcommand{\PL}{\protect\ensuremath{\mathrm{PL}}}
\newcommand{\Ptime}{\protect\ensuremath{\mathrm{P}}} 
\newcommand{\PTime}{\protect\ensuremath{\mathrm{P}}}
\newcommand{\NC}[1]{\protect\ensuremath{\mathrm{NC}^{#1}}}

\newcommand{\Ind}{\mathbb{I}}
\newcommand{\pow}[1]{2^{#1}}
\newcommand{\size}[1]{|#1|}
\newcommand{\teamsize}[1]{||#1||}
\newcommand{\set}[1]{\{#1\}}
\renewcommand{\phi}{\varphi}
\renewcommand{\epsilon}{\varepsilon}
\newcommand{\nats}[0]{\mathbb{N}}

\newcommand{\suffix}[2]{#1[#2,\infty)}
\newcommand{\var}{\mathcal{V}}
\newcommand{\hyltl}{\protect\ensuremath{\mathrm{HyperLTL}}\xspace}
\DeclareMathOperator\dep{\mathrm{dep}}

\newcommand{\kripke}{\mathcal{K}}
\newcommand{\aut}{\mathcal{A}}
\newcommand{\agreeson}[3]{#2\overset{#1}{\Leftrightarrow}#3}

\newcommand{\decisionproblem}[3]{\problemdef{#1}{#2}{#3}}
\newcommand{\myquot}[1]{``#1''}

\newcommand{\yes}{\protect\ensuremath{\checkmark}}
\newcommand{\no}{\protect\ensuremath{\times}}

\newcommand{\teamup}{\mathcal{T}}
\newcommand{\eval}[1]{\llbracket #1\rrbracket}
\newcommand{\lcm}{\mathrm{lcm}}
\newcommand{\prfx}{\mathrm{prfx}}

\newcommand\QBFVAL{\QBF\text-\problemFont{VAL}}
\newcommand{\qbfval}{\QBFVAL}

\newcommand{\dom}{\mathrm{Dom}}

\newcommand{\logicOpFont}[1]{\mathsf{#1}}

\newcommand\true\top
\newcommand\false\bot

\newcommand\PROP{\ap}

\newcommand\QBF{\problemFont{QBF}}

\newcommand{\problemdef}[3]{%
\begin{center}
\begin{tabular}{rp{11cm}}\toprule
 \textbf{Problem}:& #1\\
 \textbf{Input}:& #2\\
 \textbf{Question}:& #3\\\bottomrule
\end{tabular}
\end{center}
}

\newcommand{\AX}{\protect\ensuremath\A\X}
\newcommand{\EX}{\protect\ensuremath\E\X}
\newcommand{\AF}{\protect\ensuremath\A\F}
\newcommand{\EF}{\protect\ensuremath\E\F}

\renewcommand{\P}{\classFont{P}}

\newcommand*{\ldblbrace}{\{\mskip-5mu\{}
\newcommand*{\rdblbrace}{\}\mskip-5mu\}}
\newcommand{\mset}[1]{\protect\ensuremath{\ldblbrace #1\rdblbrace}}

\newcommand{\compat}[1]{#1\textrm{-compatible}}

\usepackage{stackengine}

\newcommand{\msube}{
    \mathrel{\stackinset{c}{.1ex}{c}{0.15ex}{$\scriptscriptstyle\boldsymbol{\#}$}{$\subseteq$}}
}

\newcommand{\meq}{
    \mathrel{\stackinset{c}{}{c}{1ex}{$\scriptscriptstyle\boldsymbol{\#}$}{$=$}}
}

\newcommand{\resref}[1]{\;\hspace{0pt plus 1filll}(\text{#1})}
\newcommand{\resrefcite}[1]{\;\hspace{0pt plus 1filll}\text{\cite{#1}}}

\makeatletter
\DeclareRobustCommand\bigop[1]{%
  \mathop{\vphantom{\sum}\mathpalette\bigop@{#1}}\slimits@
}
\newcommand{\bigop@}[2]{%
  \vcenter{%
    \sbox\z@{$#1\sum$}%
    \hbox{\resizebox{\ifx#1\displaystyle.9\else 1.4\fi\dimexpr\ht\z@+\dp\z@}{!}{$\m@th#2$}}%
  }%
}
\makeatother
\newcommand{\bigvarovee}{\DOTSB\bigop{\varovee}}

%% file: teamltl-packages.tex
\usepackage{tikz}
\usepackage{mathtools}

\usepackage{floatflt}

\usepackage{booktabs}
\usepackage{stmaryrd}
\usepackage{xspace}
\usepackage[ruled,vlined]{algorithm2e}

\usetikzlibrary{arrows}
\usetikzlibrary{decorations.pathreplacing}

\tikzset{%
dot/.style={fill,black,circle,inner sep=0mm,minimum width=2mm},
dotwhite/.style={draw,black,fill=white,circle,inner sep=0mm,minimum width=2mm},
dw/.style={draw,black,fill=white,circle,inner sep=0mm,minimum width=3.5mm},
d/.style={draw,black,fill=black,circle,inner sep=0mm,minimum width=3.5mm,text=white}
}

\usepackage{multicol}

%% file: content.tex
\title{Synchronous Team Semantics for Temporal Logics}

\author{Andreas Krebs}
\email{mail@krebs-net.de}
\affiliation{%
  \institution{University of T\"ubingen}
  \city{T\"ubingen}
  \country{Germany}
}

\author{Arne Meier}
\affiliation{%
  \institution{Institut f\"ur Theoretische Informatik, Leibniz Universit\"at Hannover}
  \city{Hannover}
  \country{Germany}}
\orcid{0000-0002-8061-5376}
\email{meier@thi.uni-hannover.de}

\author{Jonni Virtema}
\affiliation{%
  \institution{School of Computing Science, University of Glasgow}
  \city{Glasgow}
  \country{UK}
}
\orcid{0000-0002-1582-3718}
\email{j.t.virtema@sheffield.ac.uk}

\author{Martin Zimmermann}
\affiliation{
	\institution{Department of Computer Science, Aalborg University}
	\city{Aalborg}
	\country{Denmark}
}
\orcid{0000-0002-8038-2453}
\email{mzi@ca.aau.dk}

\renewcommand{\shortauthors}{Krebs et al.}

\begin{abstract}
We present team semantics for two of the most important linear and branching time specification languages, Linear Temporal Logic (LTL) and  Computation Tree Logic (CTL).

With team semantics, LTL is able to express hyperproperties, which have in the last decade been identified as a key concept in the verification of information flow properties.
	We study basic properties of the logic and classify the computational complexity of its satisfiability, path, and model checking problem.
	Further, we examine how extensions of the basic logic react to adding additional atomic operators.
	Finally, we compare its expressivity to the one of HyperLTL, another recently introduced logic for hyperproperties.
	Our results show that LTL with team semantics is a viable alternative to HyperLTL, which complements the expressivity of HyperLTL and has partially better algorithmic properties.

For CTL with team semantics, we investigate the computational complexity of the satisfiability and model checking problem. The satisfiability problem is shown to be EXPTIME-complete while we show that model checking is PSPACE-complete.
\end{abstract}



\maketitle

\section{Introduction}
Guaranteeing security and privacy of user information is a key requirement in software development. 
However, it is also one of the hardest goals to accomplish. 
One reason for this difficulty is that such requirements typically amount to reasoning about the flow of information and relating different execution traces of the system in question. 
Many of the requirements of interest are not trace properties, that is, properties whose satisfaction can be verified by considering each computation trace in isolation.
Formally, a trace property~$\phi$ is a set of traces and a system satisfies $\phi$ if each of its traces is in $\phi$. 
An example of a trace property is the property~\myquot{the system terminates eventually}, which is satisfied if every trace eventually reaches a terminating state.  
In contrast, the property~\myquot{the system terminates within a bounded amount of time} is no longer a trace property. Consider a system that has a trace~$t_n$ for every $n$, so that $t_n$ only reaches a terminating state after $n$ steps. 
This system does not satisfy the bounded termination property, but each individual trace~$t_n$ could also stem from a system that does satisfy it. 
Thus, satisfaction of the property cannot be verified by considering each trace in isolation. 
Properties with this characteristic were termed \emph{hyperproperties} by Clarkson and Schnei\-der~\cite{DBLP:journals/jcs/ClarksonS10}. 
Formally, a hyperproperty~$\phi$ is a set of sets of traces and a system satisfies $\phi$ if its set of traces is contained in $\phi$.
The conceptual difference to trace properties allows hyperproperties to specify a much richer landscape of properties including information flow properties capturing security and privacy specifications.
Furthermore, one can also express specifications for symmetric access to critical resources in distributed protocols and Hamming distances between code words in coding theory~\cite{markusPhD}. 
However, the increase in expressiveness requires novel approaches to specification and~verification. 

\paragraph*{HyperLTL} 
Trace properties are typically specified in temporal logics, most prominently in Linear Temporal Logic (\LTL)~\cite{Pnueli/1977/TheTemporalLogicOfPrograms}. 
Verification of \LTL specifications is routinely employed in industrial settings and marks one of the most successful applications of formal methods to real-life problems. 
Recently, this work has been extended to hyperproperties: \hyltl, \LTL equipped with trace quantifiers, has been introduced to specify hyperproperties~\cite{DBLP:conf/post/ClarksonFKMRS14}. 
A model of a \hyltl-formula is a set of traces and the quantifiers range over these traces. 
This logic is able to express the majority of the information flow properties found in the literature (we refer to Section~3 of~\cite{DBLP:conf/post/ClarksonFKMRS14} for a comprehensive list). 
The satisfiability problem for \hyltl is highly undecidable~\cite{FKTZ21} while the model checking problem is decidable, albeit of non-elementary complexity~\cite{DBLP:conf/post/ClarksonFKMRS14,DBLP:conf/cav/FinkbeinerRS15}. 
In view of this, the full logic is too strong for practical applications. Fortunately most information flow properties found in the literature can be expressed with at most one quantifier alternation and consequently belong to decidable (and tractable) fragments. 
Further works on $\hyltl$ have studied runtime verification~\cite{DBLP:conf/rv/BonakdarpourF16,DBLP:conf/rv/FinkbeinerHST17}, connections to first-order logic~\cite{DBLP:conf/stacs/Finkbeiner017}, provided tool support~\cite{DBLP:conf/cav/FinkbeinerRS15,FinkbeinerH16}, and presented applications to \myquot{software doping}~\cite{DBLP:conf/esop/DArgenioBBFH17} and the verification of web-based workflows~\cite{634}.
Recent works have also considered asynchronous extensions of \hyltl~\cite{DBLP:conf/cav/BaumeisterCBFS21,DBLP:conf/lics/BozzelliPS21,DBLP:journals/pacmpl/GutsfeldMO21,DBLP:conf/concur/BartocciHNC23} and verification tools for full \hyltl~\cite{DBLP:conf/tacas/BeutnerF23} (most of the previous tools were designed for fragments without quantifier alternation).
In contrast, there are natural properties, e.g., bounded termination, which are not expressible in \hyltl (which is an easy consequence of a much stronger non-expressibility result~\cite{BozzelliMP15}).

\paragraph*{Team Semantics} Intriguingly, there exists another modern family of logics which operates on sets of objects instead of objects alone. In 1997, Hodges introduced compositional semantics for Hintikka's Independence-friendly logic~\cite{Hodges97c} where the semantical object is a set of first-order assignments. Hence formulae, in this setting, define sets of sets of first-order assignments. We could call these definable sets \emph{first-order hyperproperties}, while Hodges himself called them \emph{trumps}. A decade later V\"a\"an\"anen~\cite{vaananen07} introduced \emph{Dependence logic} that adopted Hodges' semantics and reimagined Independence-friendly logic. Dependence logic extends first-order logic by atoms expressing that \myquot{the value of a variable~$x$ functionally determines the value of a variable~$y$}.
Obviously, such statements only make sense when being evaluated over a set of assignments. Therefore, they are, using the parlance introduced above, hyperproperties. 
In the language of dependence logic, such sets are called \emph{teams} and the semantics is termed \emph{team semantics}.

After the introduction of dependence logic, a whole family of logics with different atomic statements have been introduced in this framework: \emph{independence logic}~\cite{gv13} and \emph{inclusion logic}~\cite{Galliani12} being the most prominent.
Interest in these logics is rapidly growing and connections to a plethora of disciplines have been drawn, e.g., to database theory~\cite{HannulaKV20}, real valued computation~\cite{abs-2003-00644}, quantum foundations~\cite{AlbertG22,abramsky2021team}, and to the study of argumentation~\cite{DBLP:conf/foiks/MahmoodVBN24} and causation~\cite{DBLP:conf/csl/BarberoV24}.

\paragraph*{Our Contribution}
The conceptual similarities between hyperproperties and team properties raise the question \emph{what is the natural team semantics for temporal logics}. In this paper, we develop team semantics for \LTL and \CTL, obtaining the logics \TeamLTL and \TeamCTL.
We analyse their expressive power and the complexity of their satisfiability and model checking problems and subsequently compare \TeamLTL to \hyltl.
While team semantics has previously been defined for propositional and modal logic~\cite{va08}, we are the first to consider team semantics for temporal logics~\cite{kmv15}.

Our complexity results are summarised in Figure~\ref{fig:overview}.
We prove that the satisfiability problem (TSAT) for \TeamLTL is $\PSPACE$-complete by showing that the problem is equivalent to \LTL satisfiability under classical semantics. 
We consider two variants of the model checking problem. 
As there are uncountably many traces, we have to represent teams, i.e., sets of traces, in a finitary manner. 
The path checking problem (TPC) asks to check whether a finite team of ultimately periodic traces satisfies a given formula. 
As our main technical result, we establish this problem to be $\PSPACE$-complete. 
In the (general) model checking problem (TMC), a team is represented by a finite transition system. 
Formally, given a transition system and a formula, the model checking problem asks to determine whether the set of traces of the system satisfies the formula. 
We give a polynomial space algorithm for the model checking problem for the disjunction-free fragment, while we leave open the complexity of the general problem.  
Disjunction plays a special role in team semantics, as it splits a team into two. 
As a result, this operator is commonly called \emph{splitjunction}.
Note that, even for a very simple transition system, the set of traces generated by the system may be uncountable.
Hence, evaluating the splitjunction compositionally may require us to deal with infinitely many splits of an uncountable team, if a splitjunction is under the scope of a globally-operator (such as in the formula $\G(\varphi \lor \psi)$). This raises interesting language-theoretic questions. 
Additionally, we study the effects for complexity that follow when our logics are extended by dependence atoms, the contradictory negation $\sim$, the Boolean disjunction $\varovee$, and so-called \emph{generalised atoms}~\cite{kuusisto15}.
Finally, we show that  \TeamLTL is able to specify properties which are not expressible in \hyltl and \emph{vice versa}.

Recall that satisfiability for \hyltl is highly undecidable~\cite{FKTZ21} and model checking is of non-elementary complexity~\cite{markusPhD,Mascle020}.
Our results show that similar problems for \TeamLTL have a much simpler complexity while some hyperproperties are still expressible (e.g., input determinism, see page~\pageref{pg:input-nonint}, or bounded termination).
This demonstrates that \TeamLTL is a viable alternative for the specification and verification of hyperproperties that complements \hyltl.

In the second part of the paper, we develop team semantics for \CTL.
We establish that the satisfiability problem for the resulting logic~\TeamCTL is $\EXPTIME$-complete while the model checking problem is $\PSPACE$-complete. 
While \CTL satisfiability is already $\EXPTIME$-complete for classical semantics~\cite{fila79,pr80}, \CTL model checking is $\Ptime$-complete~\cite{clemsi86,sc02} for classical semantics, i.e., team semantics increases the complexity of the problem. 
For the team semantics setting, we extend our model checking result to cover also finite sets of FO-definable generalised atoms and the contradictory negation.
Finally, we compare the expressiveness of \TeamCTL and \CTL with classical semantics.

\paragraph*{Prior Work}
Preliminary versions of this work have been published in the proceedings of the 43rd International Symposium on Mathematical Foundations of Computer Science, MFCS 2018~\cite{DBLP:conf/mfcs/KrebsMV018}, and in the proceedings of the 22nd International Symposium on Temporal Representation and Reasoning, TIME 2015~\cite{kmv15}. 
The following list summarises how this article extends the conference versions:
\begin{itemize}
	\item The presentation of the article has been thoroughly revised. Moreover, we now include an overview of the subsequent works on temporal team semantics published after our conference articles.
	\item The proof of Lemma~\ref{lem:TPCs_PSPACEhard} has been considerably extended and now contains also the correctness part of the reduction.
	\item The proof of Lemma~\ref{lem:TPCs_in_PSPACE} was omitted in the conference version and is now included.
	\item The full proof of Theorem~\ref{thm:tmcs-splitfree} is now included.
	\item The formulation of Theorem~\ref{thm:GenAtoms} is slightly strengthened and a proof in now included.
	\item Theorem~\ref{thm:undef} is new. 
	\item The construction and figures in the proof of Lemma~\ref{lem:mcs-pspacehard} have been improved.
	\item The proof of Lemma~\ref{lem:smcpspace} has been corrected.
	\item Theorem~\ref{thm:tmc-ctl-d-sim} is new.
\end{itemize}
\begin{figure}
	\Description[Overview Table pointing to references and results]{This table shows an overview of the results. The first column is about the logic, the second about the problem TSAT, the third about TPC, and the last about TMC. Each row then states what the complexities of this problem is with respect to the logic in the first cell of the row. Furthermore there are either pointers to the result in the paper or to a literature source.}
	\resizebox{\textwidth}{!}{$
	\begin{array}{*{4}{l}}\toprule
		\text{Logic} &
		\text{SAT} &
		\text{PC} &
		\text{MC}\\\midrule
		\LTL & 
		\PSPACE\resrefcite{SistlaC85} & 
		\P\resrefcite{kf09,ms03} & 
		\PSPACE\resrefcite{SistlaC85} 
		\\
		\TeamLTL & 
		\PSPACE\resref{Prop.~\ref{prop:TMCa-TSATa-PSPACE}} & 
		\PSPACE\resref{Thm.~\ref{thm:tpcs}} & 
		\PSPACE\text{-hard}\resref{Thm.~\ref{thm:tpcs}} 
		\\
		\TeamLTL(\dep) & 
		\PSPACE\resref{Prop.~\ref{thm:TSATa-TSATs-dep}} & 
		\PSPACE\resref{Thm.~\ref{thm:GenAtoms}} & 
		\NEXPTIME\text{-hard}\resref{Thm.~\ref{thm:TMCa-TMCs-dep}} 
		\\
		\TeamLTL(\varovee, \mathcal D) & 
		\text{$\Sigma^0_1$-hard}\resrefcite{DBLP:conf/fsttcs/VirtemaHFK021} &
		\PSPACE\resref{Thm.~\ref{thm:GenAtoms}} & 
		\text{$\Sigma^0_1$-hard}\resrefcite{DBLP:conf/fsttcs/VirtemaHFK021}
		\\
		\TeamLTL(\mathcal D,\sim) &
		\text{third-order arithmetic} \resrefcite{DBLP:journals/tcs/Luck20}\footnotemark[1] &
		\PSPACE\resref{Thm.~\ref{thm:GenAtoms}} &
		\text{third-order arithmetic} \resrefcite{DBLP:journals/tcs/Luck20}\footnotemark[1]
		\\
		\TeamLTL-\lor & 
		? & 
		? &
		\in\PSPACE\resref{Thm.~\ref{thm:tmcs-splitfree}}  
		\\
		\midrule
		\CTL & 
		\EXPTIME\resrefcite{fila79,pr80} & 
		\text{---} & 
		\P\resrefcite{clemsi86,sc02} 
		\\
		\TeamCTL & 
		\EXPTIME\resref{Thm.~\ref{thm:SAT}} & 
		\text{---} & 
		\PSPACE\resref{Thm.~\ref{thm:TMC(CTL)-c}} 
		\\
		\TeamCTL(\mathcal D,{\sim}) & 
		? & 
		\text{---} & 
		\PSPACE\resref{Thm.~\ref{thm:tmc-ctl-d-sim}} 
		\\
		\bottomrule
	\end{array}
	$}
	\caption{Overview of complexity results for \TeamLTL and \TeamCTL. `$\dep$' refers to dependence atoms, `$\sim$' refers to the contradictory negation, $\varovee$ refers to the Boolean disjunction, $\mathcal D$ refers to any finite set of first-order definable generalised atoms, and `$\TeamLTL-\lor$' refers to disjunction free \TeamLTL. All results are completeness results unless otherwise stated. We write `?' for open cases and `--' if the problem is not meaningful for the logic. Results concerning $\sim$ remain true for extensions with connectives that are expressible with at most polynomial blow-up. In particular, $\varphi \varovee \psi$ is equivalent with $\sim(\sim\varphi \land \sim\psi)$.}\label{fig:overview}
\end{figure}
\footnotetext[1]{\citeauthor{DBLP:journals/tcs/Luck20} \cite{DBLP:journals/tcs/Luck20} proves this result without generalised atoms. However, since all generalised atoms are expressible in $\TeamLTL(\sim)$ with double exponential blow-up in the formula size (see \cite[Proposition 17]{DBLP:journals/corr/abs-2010-03311} and \cite[page 6]{DBLP:journals/tcs/Luck20}), the results carries over to the extension of $\TeamLTL(\sim)$ with generalised atoms.}

\paragraph*{Related work on TeamLTL}
Since the publication of the original conference articles~\cite{kmv15,DBLP:conf/mfcs/KrebsMV018} several follow-up works have been published. 
\citeauthor{DBLP:conf/mfcs/KrebsMV018}~\cite{DBLP:conf/mfcs/KrebsMV018} also introduced an asynchronous variant of \TeamLTL. 
This line of work has been continued by Kontinen~et~al.~\cite{DBLP:conf/wollic/KontinenS21,DBLP:conf/mfcs/KontinenSV23,DBLP:conf/foiks/KontinenSV24}. 
In particular, \citeauthor{DBLP:conf/wollic/KontinenS21}~\cite{DBLP:conf/wollic/KontinenS21} considered the extension of \TeamLTL with the contradictory negation and studied its complexity and translations to first-order and second-order logic. 
Then, \citeauthor{DBLP:conf/mfcs/KontinenSV23}~\cite{DBLP:conf/mfcs/KontinenSV23} fixed a mistake in the original definition of asynchronous \TeamLTL~\cite{DBLP:conf/mfcs/KrebsMV018}, introduced a novel set-based semantics for \TeamLTL and studied its complexity and expressivity. 
Hence, for the definition of asynchronous \TeamLTL, please refer to the work of \citeauthor{DBLP:conf/mfcs/KontinenSV23}~\cite{DBLP:conf/mfcs/KontinenSV23}. 
In their recent work on this topic, \citeauthor{DBLP:conf/foiks/KontinenSV24}~\cite{DBLP:conf/foiks/KontinenSV24} revealed a tight connection between set-based asynchronous \TeamLTL and the one-variable fragment of $\hyltl$.

\citeauthor{DBLP:journals/tcs/Luck20}~\cite{DBLP:journals/tcs/Luck20} showed that the satisfiability and model checking problems of synchronous \TeamLTL with contradictory negation are complete for third-order arithmetic.
\citeauthor{DBLP:conf/fsttcs/VirtemaHFK021}~\cite{DBLP:conf/fsttcs/VirtemaHFK021} studied the expressivity and complexity of various extensions of synchronous \TeamLTL. 
In particular, they identified undecidable cases and cases when team logics can be translated to HyperQPTL and HyperQPTL$^+$. 
By doing so, they mapped the undecidability landscape of synchronous \TeamLTL.

\citeauthor{DBLP:conf/lics/GutsfeldMOV22}~\cite{DBLP:conf/lics/GutsfeldMOV22} introduced a flexible team-based formalism to logically specify asynchronous hyperproperties based on so-called \emph{time evaluation functions}, which subsumes synchronous \TeamLTL and can be used to model diverse forms of asynchronicity. 

\citeauthor{DBLP:journals/tocl/BellierBMM23}~\cite{DBLP:journals/tocl/BellierBMM23} introduced a team-based formalism for quantified propositional temporal logic that takes inspiration from first-order team logics which hence is orthogonal to \TeamLTL.

Finally, while there exists a deep understanding of the complexity of the satisfiability and model checking problems for classical $\LTL$~\cite{DBLP:journals/tocl/BaulandM0SSV11,DBLP:journals/corr/abs-0812-4848} and $\CTL$~\cite{mmtv09,mmmv12,DBLP:journals/acta/KrebsMM19}, these investigations do neither directly nor generally carry over to team semantics.

Finally, the complexity of the most important verification problems for \hyltl and related hyperlogics has been settled~\cite{DBLP:conf/cav/FinkbeinerRS15,markusPhD,FinkbeinerH16,FKTZ21,DBLP:conf/csl/Frenkel025,fragments,regaud2024complexityhyperqptl}.

\section{Preliminaries}\label{sec:prelims}

The non-negative integers are denoted by $\nats$ and the power set of a set~$S$ is denoted by $\pow{S}$. Throughout the paper, we consider a countably infinite set~$\ap$ of atomic propositions.

\paragraph*{Computational Complexity}
We will make use of standard notions in complexity theory. 
In particular, we will use the complexity classes $\PTime$, $\PSPACE$, $\EXPTIME$, and $\NEXPTIME$. 

An alternating Turing machine (ATM) is a non-deterministic Turing machine whose state space is partitioned into two types of states: existential and universal. 
Acceptance for ATMs is defined in an inductive way on any computation tree with respect to a given input as follows.
A halting configuration is accepting if and only if it contains an accepting state. 
An `inner' configuration (which is not a halting configuration) is accepting depending on its type of state: if it is existential then it is accepting if at least one of its children is accepting, and if it is universal then it is accepting if all of its children are accepting. 
An ATM accepts its input if and only if its initial configuration is accepting.
Finally, the complexity class $\ATIME(x,y)$ is the set of problems solvable by ATMs in a runtime of $x$ with $y$-many alternations. 

Most reductions used in the paper are $\leqpm$-reductions, that is, polynomial time, many-to-one reductions.

\paragraph*{Traces and Teams}
 A \emph{trace} over $\ap$ is an infinite sequence from $ (\pow{\ap})^\omega$; a finite trace is a finite sequence from $(\pow{\ap})^*$. 
 The length of a finite trace~$t$ is denoted by $\size{t}$. The empty trace is denoted by $\epsilon$ and the concatenation of a finite trace~$t_0$ and a finite or infinite trace~$t_1$ by $t_0t_1$.
 Unless stated otherwise, a trace is always assumed to be infinite. 
 
 A \emph{team} is a (potentially infinite) set of traces.
 Given a trace~$t = t(0) t(1) t(2) \cdots$ and $i \ge 0$, we define $t[i,\infty) \dfn t(i) t(i+1) t(i+2) \cdots$, which we lift to teams~$T \subseteq (\pow{\ap})^\omega$ by defining $T[i,\infty) \dfn \set{\, t[i,\infty) \mid t \in T\, }$. 
 
 To serve as inputs for algorithms, we need to represent traces, infinite sequences of possibly infinite sets, finitely.
 A trace~$t$ is \emph{ultimately periodic} if it is of the form~$t = t_0 \cdot t_1^\omega = t_0 t_1 t_1 t_1 \cdots$ for two finite traces~$t_0$ and $t_1$ with $\size{t_1}>0$. 
 As a result, an ultimately periodic trace $t$ is represented by the pair~$(t_0, t_1)$; we define $\eval{(t_0, t_1)} = t_0t_1^\omega$. 
 Given a set~$\teamup$ of such pairs, we define $\eval{\teamup} = \set{\,\eval{(t_0,t_1)} \mid (t_0, t_1) \in \teamup\,}$, which is a team of ultimately periodic traces. 
 We call $\teamup$ a team encoding of $\eval{\teamup}$. 
 A team encoding is finite, if it contains finitely many elements~$(t_0,t_1)$ where each letter of $t_0$ and each letter of $t_1$ is a finite subset of $\ap$.
This implies that a finite team encoding encodes a finite team, which can be represented finitely. 
The \emph{representation size} of a finite team encoding $\teamup$, denoted by $\teamsize{\teamup}$, is defined as
\[
	\sum_{(t_0,t_1)\in\teamup} \size{t_0} + \size{t_1} + 2,
\]
where the `$+2$' accounts for adequate separator symbols.

\paragraph*{Linear Temporal Logic}
The formulae of Linear Temporal Logic (\LTL)~\cite{Pnueli/1977/TheTemporalLogicOfPrograms} are defined via the grammar
\[
\varphi \ddfn 
p\mid 
\lnot p\mid
\varphi\land\varphi\mid
\varphi\lor\varphi\mid 
\X\varphi\mid 
\varphi\U\varphi\mid 
\varphi\R\varphi,
\]
where $p$ ranges over the atomic propositions in $\ap$. 
We define the following usual shorthands: $\top\coloneqq p\lor \lnot p$, $\bot\coloneqq p\land \lnot p$, $\F\varphi\coloneqq \top\U\varphi$, $\G\varphi\coloneqq \bot\R\varphi$.
The length of a formula is defined to be the number of Boolean and temporal connectives occurring in it. 
Often, the length of an \LTL-formula is defined to be the number of syntactically different subformulae, which might be exponentially smaller. 
Here, we need to distinguish syntactically equal subformulae which becomes clearer after defining the semantics.
As we only consider formulae in negation normal form, we use the full set of temporal operators.

Next, we recall the classical semantics of \LTL before we introduce team semantics. 
For traces~$t \in (\pow{\ap})^\omega$ we define the following:
\begin{tabbing}
    $t \ltlmodels  \psi\lor\phi$ \= if \= Rechts \kill
	$t \ltlmodels  p$\> if \>$p\in t(0)$,\\
	$t \ltlmodels  \lnot p$\> if \>$ p\notin t(0)$,\\
	$t \ltlmodels  \psi\land\phi$\> if \>$ t \ltlmodels \psi \text{ and }t \ltlmodels \phi$,\\
	$t \ltlmodels  \psi\lor\phi$\> if \>$ t \ltlmodels \psi \text{ or }t \ltlmodels \phi$,\\
	$t \ltlmodels \X\varphi$\> if \>$t[1,\infty)\ltlmodels\varphi$,\\
	$t\ltlmodels \psi\U\phi$\> if\> $\exists k \ge 0$ such that $t[k,\infty)\ltlmodels \phi$ and $\forall k' < k$ we have that $t[k',\infty)\ltlmodels \psi$, and\\ 
	$t\ltlmodels \psi\R\phi$\> if\> $\forall k\ge 0$ we have that $t[k,\infty)\ltlmodels \phi$ or $\exists k' < k$ such that $t[k',\infty)\ltlmodels \psi$.
\end{tabbing}

\section{Team Semantics for LTL}\label{sec:TeamLTL}
Next, we introduce (synchronous) team semantics for \LTL, obtaining the logic \TeamLTL.
The syntax of \TeamLTL is the same as the one of \LTL, but the semantics differs:
For teams $T \subseteq (\pow{\ap})^\omega$ we define the following:
\begin{tabbing}
    $T\models \psi\land\phi$ \= if \= Rechts \kill
	$T\models p$\> if \> $\forall t \in T$ we have that $p\in t(0)$,\\
	$T\models \lnot p$\> if \> $\forall t \in T$ we have that $p\notin t(0)$,\\
	$T\models \psi\land\phi$\> if \>$T\models\psi$  and $T\models\phi$,\\
	$T\models \psi\lor\phi$\> if \> $\exists T_1\cup T_2=T$ such that $T_1\models\psi$ and $T_2\models\phi$,\\
	$T\models\X\varphi$\> if \>$T[1,\infty)\models\varphi$,\\
	$T\models\psi\U\phi$\> if \> $\exists k\ge 0$ such that $T[k,\infty)\models\phi$ and $\forall k' <k$ we have that $T[k',\infty)\models\psi$, and\\
	$T\models\psi\R\phi$\> if \> $\forall k \ge 0$ we have that $T[k,\infty)\models\phi$ or $\exists k'<k$ such that $T[k',\infty)\models\psi$.
\end{tabbing}

We call expressions of the form $\psi\lor\phi$ \emph{splitjunctions} to emphasise that in team semantics disjunction splits a team into two parts.
Similarly, the $\lor$-operator is referred to as a \emph{splitjunction}.
Notice that the object (a trace or a team) left of $\models$ determines which of the above semantics is used (classical or team semantics). Some subsequent works consider team semantics over multisets of traces~\cite{DBLP:conf/lics/GutsfeldMOV22}. For synchronous \TeamLTL (as introduced here) the generalisation to multisets would allow the implementation of new quantitative dependency statements as atomic formulae. From the above semantics only disjunction would need to be reinterpreted via disjoint (multiset) unions. Since having multiset teams would not have any meaningful impact on \TeamLTL, we adopt the slightly simpler set-based semantics. 

\begin{figure}
	\Description[Properties Overview Table]{Left column property, right column definition.}
	\begin{center}
\scalebox{0.85}{
 \begin{tabular}{llc}\toprule
 	property & definition & \\\midrule
 	empty team property & $\emptyset\models\phi$ & \yes\\
 	downward closure & $T\models\phi$ implies $\forall T'\subseteq T$: $T'\models\phi$ & \yes \\
 	union closure & $T\models\phi$ and $T'\models\phi$ implies $T\cup T'\models\phi$ & \no\\
 	flatness & $T\models\phi$ if and only if $\forall t\in T$: $\{t\}\models\phi$	& \no\\
 	singleton equivalence & $\{t\}\models\phi$ if and only if $t\ltlmodels\phi$ & \yes \\\bottomrule
 \end{tabular}
}
\end{center}
\caption{Structural properties overview for \TeamLTL.}\label{fig:structuralpropertiesoverview}
\end{figure}

\begin{example}
If $p$ is an atomic proposition encoding that a computation has ended, then $\F p$ defines the hyperproperty \emph{bounded termination}. In particular, a possibly infinite team $T$ satisfies $\F p$ if there is a natural number $n\in\mathbb{N}$ such that $p\in t(n)$, for every $t\in T$.
\end{example}

We consider several standard properties of team semantics (cf., e.g.~\cite{DKV16}) and verify which of these hold for our semantics for \LTL. 
These properties are later used to analyse the complexity of the satisfiability and model checking problems. 
See Figure~\ref{fig:structuralpropertiesoverview} for the definitions of the properties and a summary for which of the properties hold for our semantics. 
The positive results follow via simple inductive arguments. 
For the fact that team semantics is not union closed, consider teams $T = \set{\set{p} \emptyset^\omega}$ and $T' = \set{\emptyset\set{p} \emptyset^\omega}$. 
Then, we have $T \models \F p$ and $T'\models \F p$ but $T \cup T' \not\models \F p$. 
This also implies that flatness fails, for it implies union closure.

A Kripke structure~$\kripke = (W, R, \eta, w_I)$ consists of a finite set~$W$ of worlds, a left-total transition relation~$R \subseteq W \times W$, a labeling function~$\eta \colon W \rightarrow \pow{\ap}$ such that $\eta(w)$ is finite for  each $w \in W$, and an initial world~$w_I \in W$. 
A path~$\pi$ through $\kripke$ is an infinite sequence~$\pi = \pi(0) \pi(1) \pi(2) \cdots \in W^\omega$ such that $\pi(0) = w_I$ and $(\pi(i), \pi(i+1)) \in R$ for every $i \ge 0$. 
The trace of $\pi$ is defined as $t(\pi) = \eta(\pi(0))\eta(\pi(1))\eta(\pi(2)) \cdots \in (\pow{\ap})^\omega$. 
A Kripke structure~$\kripke$ induces the team~$T(\kripke) = \set{\,t(\pi) \mid \pi \text{ is a path through }\kripke\,}$. 

Next, we define the most important verification problems for \TeamLTL, namely satisfiability and two variants of the model checking problem. 
For classical \LTL, one studies the path checking problem and the model checking problem. The difference between these two problems lies in the type of structures one considers. 
Recall that a model of an \LTL-formula is a single trace. In the path checking problem, a trace~$t$ and a formula~$\phi$ are given, and one has to decide whether $t \ltlmodels \phi$. 
This problem has applications to runtime verification and monitoring of reactive systems~\cite{kf09,ms03}. 
In the model checking problem, a Kripke structure~$\kripke$ and a formula~$\phi$ are given, and  one has to decide whether every execution trace~$t$ of $\kripke$ satisfies $\phi$.

The satisfiability problem of \TeamLTL is defined as follows.
 \decisionproblem{$\TSAT(\TeamLTL)$ --- \TeamLTL satisfiability.}{A \TeamLTL-formula $\phi$.}{Is there a non-empty team~$T$ such that $T\models\phi$?} 
The non-emptiness condition is necessary, as otherwise every formula is satisfiable due to the empty team property (see Figure~\ref{fig:structuralpropertiesoverview}).
Also note that, due to downward closure (see Figure~\ref{fig:structuralpropertiesoverview}), a \TeamLTL-formula $\phi$ is satisfiable if and only if there is a singleton team that satisfies the formula. 
From this and singleton equivalence (see Figure~\ref{fig:structuralpropertiesoverview}), we obtain the following result from the identical result for \LTL under classical semantics~\cite{SistlaC85}.

\begin{proposition}\label{prop:TMCa-TSATa-PSPACE}\label{thm:TPCa}
$\TSAT(\TeamLTL)$ is $\PSPACE$-complete w.r.t.\ $\leqpm$-reductions.
\end{proposition}

We consider the generalisation of the path checking problem for \LTL (denoted by $\LTLPC$), which asks for a given ultimately periodic trace~$t$ and a given formula~$\phi$, whether $t \ltlmodels\phi$ holds. 
In the team semantics setting, the corresponding question is whether a given finite team comprised of ultimately periodic traces satisfies a given formula. 
Such a team is given by a team encoding~$\teamup$. To simplify our notation, we will write $\teamup\models \varphi$ instead of $\eval{\teamup} \models \varphi$.

 \decisionproblem{$\TPC(\TeamLTL)$ --- \TeamLTL Path Checking.}{A \TeamLTL-formula $\phi$ and a finite team encoding~$\teamup$.}{Does $\teamup\models\phi$?}

Now, consider the generalised model checking problem where one checks whether the team of traces of a Kripke structure satisfies a given formula. 
This is the natural generalisation of the model checking problem for classical semantics, denoted by $\LTLMC$, which asks, for a given Kripke structure~$\kripke$ and a given \LTL-formula~$\phi$, whether $t \ltlmodels \phi$ for every trace $t$ of $\kripke$.

  \decisionproblem{$\TMC(\TeamLTL)$ --- \TeamLTL Model Checking.}{A \TeamLTL-formula $\phi$ and a Kripke structure $\kripke$.}{Does $T(\kripke)\models\phi$?}

\section{Complexity Results for TeamLTL}

Next, we examine the computational complexity of path and model checking with respect to team semantics.

\subsection{Path Checking}
The following problem for quantified Boolean formulae (qBf) is well-known~\cite{lad77,st77} to be $\PSPACE$-complete:
\problemdef{$\qbfval$ --- Validity problem for quantified Boolean formulae.}{A quantified Boolean formula $\varphi$.}{Is $\varphi$ valid?}

\begin{lemma}\label{lem:TPCs_PSPACEhard}
	$\TPC(\TeamLTL)$ is $\PSPACE$-hard w.r.t.\ $\leqpm$-reductions.
\end{lemma}
\begin{proof}
	\begin{figure}
		\Description[Trace overview]{U(i) is the trace which is a cycle of 6 states, where after the start the first state has qi and dollar symbol labeled, the second state has only dollar symbol, the third nothing, the fourth dollar symbol, and the last has qi dollar symbol and hash tag symbol. E(i) are two cycles of length 3. They differ in the second, resp. third state. The second for E(i) is called T(i,1) and has xi, qi and dollar symbol labeled while the other trace T(i,0) has instead a dollar symbol. In the last state, for T(i,1) we have a dollar symbol and a hash tag symbol, while for T(i,0) we have xi, qi, dollar symbol, hash tag symbol. The last image is depending on which literal is in the clause. If it is a positive one then the trace L(j,k) for the jth literal in clause k is a cycle of length 3 and contains at the second state xi and a dollar symbol. The third state has a dollar and a hash tag symbol. If the literal is negative, then the last two states are swapped.}
		\centering
		\begin{tikzpicture}[x=1cm,y=1.2cm,thick]
		
		\node at (0.5,.75) {$U(i)$};
		
			\foreach \n/\x/\y/\pos/\lab/\col in {1/0/0/180//dotwhite,2/0/-1/180/$\substack{q_i\\\$}$/dot,3/0/-2/180/$\substack{\$}$/dot,4/1/0/0//dot,5/1/-1/0/$\substack{\$}$/dot,6/1/-2/0/$\substack{q_i\\\$\\\#}$/dot}{
				\node[\col,label={\pos:\lab}] (\n) at (\x,\y) {};
			}
			
			\foreach \f/\t in {1/2,2/3,3/4,4/5,5/6}{
				\path[-latex,black] (\f) edge (\t);
			}
			\path[-latex,black] (6) edge[out=55,in=45,looseness=1.5] (1);
		
			\node at (0,-3.35) {};
		\end{tikzpicture}\hfill
		\begin{tikzpicture}[x=1cm,y=1.2cm,thick]
		
		\node at (.5,1.3) {$E(i)$};
		
		\draw [decorate,decoration={brace,amplitude=5pt,raise=4pt}]
				(-.6,.7) -- (1.8,.7);
		
		\node at (0,.55) {$T(i,1)$};
		\node at (1.25,.55) {$T(i,0)$};
		
			\foreach \n/\x/\y/\pos/\lab/\col in {1/0/0/180//dotwhite,2/0/-1/180/$\substack{x_i\\q_i\\\$}$/dot,3/0/-2/180/$\substack{\$\\\#}$/dot,4/1/0/0//dotwhite,5/1/-1/0/$\substack{\$}$/dot,6/1/-2/0/$\substack{x_i, q_i\\\$, \#}$/dot}{
				\node[\col,label={\pos:\lab}] (\n) at (\x,\y) {};
			}
			
			\foreach \f/\t in {1/2,2/3,4/5,5/6}{
				\path[-latex,black] (\f) edge (\t);
			}
			\path[-latex,black] (6) edge[out=45,in=45,looseness=1] (4);
			\path[-latex,black] (3) edge[out=45,in=45,looseness=1] (1);
			
			\node at (0,-3.35) {};
		\end{tikzpicture}\hfill
			\begin{tikzpicture}[x=2.5cm,y=1.2cm,thick]
				\node[text width=2.5cm,align=left] at (.15,.75) {if $\ell_{jk}=x_i$,\newline then $L({j,k}):$};
				\node[dotwhite] (1) at (0,0) {};
				\node[dot,label={180:$\substack{x_i\\\$}$}] (2) at (0,-1) {};
				\node[dot,label={180:$\substack{\$\\\#}$}] (3) at (0,-2) {};
				
				\node at (0.65,.3) {\footnotesize Pos.:};
				\node at (0.65,0) {1};
				\node at (0.65,-1) {2};
				\node at (0.65,-2) {3};

				\node[text width=2.5cm,align=left] at (1.15,.75) {if $\ell_{jk}=\neg x_i$,\newline then $L({j,k}):$};
				\node[dotwhite] (a) at (1,0) {};
				\node[dot,label={0:$\substack{\$}$}] (b) at (1,-1) {};
				\node[dot,label={0:$\substack{x_i\\\$\\\#}$}] (c) at (1,-2) {};

				\foreach \from/\to in {1/2,2/3,a/b,b/c}{
					\path[draw,-latex] (\from) edge (\to);
				}
				
				\path[-latex,black] (c) edge[out=55,in=25,looseness=1] (a);
				\path[-latex,black] (3) edge[out=55,in=25,looseness=1] (1);
		
				\draw [decorate,decoration={brace,amplitude=5pt,raise=4pt,mirror}]
				(1.4,-2.1) -- (2.1,-2.1) node[xshift=-.875cm,yshift=-.2cm,below,align=center,text width=2.4cm] {$c_j$ at positions $\{1,2,3\}\setminus\{k\}$ for both $L(j,k)$};
				

				\node at (1.5,.3) {$\scriptstyle k=1$};
				\node at (1.5,-1) {$c_j$};
				\node at (1.5,-2) {$c_j$};
				
				\node at (1.75,.3) {$\scriptstyle k=2$};
				\node at (1.75,0) {$c_j$};
				\node at (1.75,-2) {$c_j$};

				\node at (2,.3) {$\scriptstyle k=3$};
				\node at (2,0) {$c_j$};
				\node at (2,-1) {$c_j$};
				
			\end{tikzpicture}
		
		\caption{Traces for the reduction presented in the proof of Lemma~\ref{lem:TPCs_PSPACEhard}.}\label{fig:tracegadgets}
		\end{figure}
	To prove the claim of the lemma, we will show that $\qbfval\leqpm \TPC(\TeamLTL)$.
Given a quantified Boolean formula $\varphi$, we stipulate, w.l.o.g., that $\varphi$ is of the form $\exists x_1\forall x_2\cdots Q x_n\chi$, where $\chi=\bigwedge_{j=1}^m\bigvee_{k=1}^3\ell_{jk}$, $Q\in\{\exists,\forall\}$, and
 $x_1,\dots,x_n$ are exactly the free variables of $\chi$ and pairwise distinct. 
 Notice that, depending on the choices made by the quantifiers, each of the $m$ clauses of $\chi$ has to contain at least one true literal in order for $\varphi$ to be valid. 
 This will be mimicked in the reduction by particular propositions $c_1,\dots,c_m$.

In the following, we define a reduction which is composed of two functions $f$ and $g$.
Given a qBf~$\varphi$, the function $f$ will define an \LTL-formula and $g$ will define a team such that $\varphi$ is valid if and only if $g(\varphi)\models f(\varphi)$.
Essentially, the team $g(\varphi)$  will contain three kinds of traces, see Figure~\ref{fig:tracegadgets}, that will be used to encode variables $i\in\{1,\dots,n\}$, clauses $j\in\{1,\dots,m\}$, and literal positions $k\in\{1,2,3\}$:
\begin{enumerate}
	\item traces which are used to mimic universal quantification ($U(i)$ and $E(i)$), 
	\item traces that are used to simulate existential quantification ($E(i)$), and 
	\item traces used to encode the matrix of $\varphi$ ($L(j,k)$).
\end{enumerate} Moreover the trace $T(i,1)$ ($T(i,0)$, resp.) is used inside the proof to encode an assignment that maps the variable $x_i$ to true (false, resp.).
Note that, $U(i), T(i,1), T(i,0), L(j,k)$ are technically singleton sets of traces.
For convenience, we identify them with the traces they contain.

Next, we inductively define the reduction function $f$ that maps qBfs to \LTL-formulae:
\begin{align*}
	f(\chi) &\coloneqq \bigvee_{i=1}^n \F x_i\lor\bigvee_{i=1}^m\F c_i, 
	\intertext{ where $\chi$ is the 3CNF-formula $\bigwedge_{j=1}^m\bigvee_{k=1}^3\ell_{jk}$ with free variables $x_1,\dots, x_n$,}
	f(\exists x_i \psi) &\coloneqq (\F q_i)\lor f(\psi),\qquad
	f(\forall x_i \psi) \coloneqq \bigl(\$ \lor (\lnot q_i\U q_i)\lor \F[\#\land \X f(\psi)]\bigr)\U \#.
\end{align*}
Intuitively, the idea is that the left part of the formula~$f(\chi)$ is used to satisfy one literal per clause and the right part is to take care of the remaining literals in a trivial way.

The reduction function $g$ that maps qBfs to teams is defined as follows with respect to the traces in Figure~\ref{fig:tracegadgets}:
\begin{align*}
	g(\chi) &\coloneqq \bigcup_{j=1}^m\{L(j,1)\cup L(j,2)\cup L(j,3)\},\\ 
	g(\exists x_i \psi) &\coloneqq E(i)\cup g(\psi),\qquad
	g(\forall x_i \psi) \coloneqq U(i)\cup E(i)\cup g(\psi).
\end{align*}

In Figure~\ref{fig:tracegadgets}, the first position of each trace is marked with a white circle.
For instance, the trace of $U(i)$ is then encoded via $$(\varepsilon,\emptyset\{q_i,\$\}\{\$\}\emptyset\{\$\}\{q_i,\$,\#\}).$$
The reduction function showing $\qbfval\leqpm \TPC(\TeamLTL)$ is then $\varphi\mapsto\langle g(\varphi),f(\varphi)\rangle$.
Clearly $f(\varphi)$ and  $g(\varphi)$ can be computed in linear time with respect to $|\varphi|$.

Intuitively, for the existential quantifier case, the formula $(\F q_i)\lor f(\psi)$ allows to continue in $f(\psi)$ with  exactly one of $T(i,1)$ or $T(i,0)$.
If $b\in\{0,1\}$ is a truth value then selecting $T(i,b)$ in the team is the same as setting $x_i$ to $b$.
For the case of $f(\forall x_i\psi)$, the formula $(\lnot q_i\U q_i)\lor \F[\#\land \X f(\psi)]$ with respect to the team $(U(i)\cup E(i))[0,\infty)$ is similar to the existential case choosing $x_i$ to be $1$ whereas for $(U(i)\cup E(i))[3,\infty)$ one selects $x_i$ to be $0$.
The use of the until operator in combination with $\$$ and $\#$ then forces both cases to happen.

Let $\varphi'=Q'x_{n'+1}\cdots Q x_n\chi$, where $Q',Q\in\{\exists,\forall\}$ and let $I$ be an assignment of the variables in $\{x_1,\dots,x_{n'}\}$ for $n'\leq n$. 
Then, let
$$g(I,\varphi') \coloneqq g(\varphi')\cup\{\, T(i,I(x_i)) \mid x_i \in \dom(I)\, \}.$$
We claim $I\models\varphi'$ if and only if $g(I,\varphi')\models f(\varphi')$.

Note that when $\varphi'= \varphi$ it follows that $I=\emptyset$ and that $g(I,\varphi')=g(\varphi)$. 
Accordingly, the lemma follows from the correctness of the claim above.
The claim is proven by induction on the number of quantifier alternations in $\varphi'$. 
\paragraph*{Induction Basis.}
$\varphi'=\chi$, this implies that $\varphi'$ is quantifier-free and $\dom(I)=\{x_1,\dots,x_n\}$. 
		
``$\Leftarrow$'': 
Let $g(I,\varphi')=T_1\cup T_2$ s.t.\ $T_1\models\bigvee_{i=1}^n \F x_i$ and $T_2\models\bigvee_{i=1}^m \F c_i$. We assume w.l.o.g.\ $T_1$ and $T_2$ to be disjoint, which is possible due to downward closure. 
We then have that \[T_2\subseteq\{\,L(j,k)\mid 1\leq j\leq m, 1\leq k \leq 3\,\}\] and
\[T_1= (\{\,L(j,k)\mid 1\leq j\leq m, 1\leq k \leq 3\,\}\setminus T_2)\cup\{\,T(i,I(x_i))\mid 1\leq i\leq n\,\}.\]
%
Due to construction of the traces, $L(j,k)\in T_2$ can only satisfy the subformula $\F c_{j'}$ for $j' = j$.
Moreover, note that there exists no $s\in\N$ such that $L(j,k)(s)\ni c_j$ for all $1\leq k\leq 3$; hence $\{L(j,1),L(j,2),L(j,3)\}$ falsifies $\F c_j$.
These two combined imply that $T_2\not\supseteq\{L(j,1),L(j,2),L(j,3)\}$, for each $1\leq j\leq m$.
However, for each $1\leq j\leq m$, any two of $L(j,k)$, $1\leq k\leq 3$, can belong to $T_2$ and hence exactly one belongs to $T_1$ (due to downward closure). 
This is true because it is impossible to satisfy all three $L(j,k)$ per clause $j$ simultaneously. 
One needs to decide for one literal per clause that is ``the'' satisfied one and is contained in $T_1$. 
The remaining literals are in $T_2$ (also if one or both are satisfied by the assignment).

Now, let $T_1=T_1^1\cup\cdots\cup T_1^n$ such that $T_1^i\models \F x_i$.
Note that
$\F x_i$ can be satisfied by $T(i',I(x_{i'}))$ only for $i' = i$. Since $T_1\supseteq\{\,T(i,I(x_i))\mid 1\leq i\leq n\,\}$, it follows that $T(i,I(x_{i})) \in T_1^i$, for each $1\leq i \leq n$.
Note also that, if $L(j,k)\in T_1$, it has to be in $T_1^i$ where $x_i$ is the variable of $\ell_{jk}$.
By construction of the traces, if $T(i,1)\in T_1^i$, we have $T_1^i(1)\models x_i$ and if $T(i,0)\in T_1^i$ then $T_1^i(2)\models x_i$. 
Thus, by construction of the traces $L(j,k)$, if $L(j,k)\in T_1$, then $I\models\ell_{jk}$.
Since, for each $1\leq j\leq m$, there is a $1\leq k\leq 3$ such that $L(j,k)\in T_1$ it follows that $I\models\varphi'$.

``$\Rightarrow$'': 
Now assume that $I\models\varphi'$.
As a result, pick for each $1\leq j\leq m$ a \emph{single} $1\leq k\leq 3$ such that $I\models\ell_{jk}$.
Denote this sequence of choices by $k_1,\dots,k_m$.
Choose $g(I,\varphi')=T_1\cup T_2$ as follows:
\begin{align*}
	T_1 &\coloneqq \{\,L(j,k_j)\mid  1\leq j\leq m, k_j\in\{1,2,3\} \}\cup\{T(i,I(x_i)\mid 1\leq i\leq n\,\}\\
	T_2 &\coloneqq \{\,L(j,1),L(j,2),L(j,3)\mid 1\leq j\leq m\,\}\setminus T_1
\end{align*}
Then $T_2\models \bigvee_{j=1}^m \F c_j$, for exactly two traces per clause are in $T_2$, and we can divide $T_2=T_2^1\cup\cdots\cup T_2^m$ where
\[
T_2^j \coloneqq \{\,L(j,k),L(j,k')\mid k,k'\in\{1,2,3\}\setminus\{k_j\}\,\},
\]
and, by construction of the traces, $T_2^j\models \F c_j$, for all $1\leq j\leq m$.
Furthermore, note that $T_1=T_1^1\cup\cdots \cup T_1^n$, where
\[
T_1^i \dfn \{\,L(j,k_j)\mid 1\leq j\leq m, I(x_i)\models\ell_{j{k_j}}\,\}\cup\{T(i,I(x_i))\}.
\]
Note that the left set of the union is a singleton. 
Now, for all $L(j,k_j) \in T_1^i$, there are two possibilities:
\begin{itemize}
	\item $I(x_i)=1$: then $x_i\in(L(j,k_j)(1)\cap T(i,I(x_i))(1))$.
	\item $I(x_i)=0$: then $x_i\in(L(j,k_j)(2)\cap T(i,I(x_i))(2))$. 
\end{itemize}
In both cases, $T_1^i\models \F x_i$, and thus $T_1\models\bigvee_{i=1}^n \F x_i$.
Hence it follows that $g(I,\varphi')\models f(\varphi')$ and the induction basis is proven.
\paragraph*{Induction Step.} 
		``Case $\varphi'=\exists x_i\psi$.'' 
		We show that $I\models \exists x_i\psi$ if and only if $g(I,\exists x_i\psi)\models f(\exists x_i\psi)$. 
		
		First note that $g(I,\exists x_i\psi)\models f(\exists x_i\psi)$ iff $E(i)\cup g(\psi)\cup\{\, T(i,I(x_i)) \mid x_i \in \dom(I)\, \}\models (\F q_i)\lor f(\psi)$, by the definitions of $f$ and $g$.
		Clearly, $E(i)\not\models \F q_i$, but both $T(i,1)\models\F q_i$ and $T(i,0)\models\F q_i$.
		Observe that $E(i)=\{T(i,1),T(i,0)\}$ and $q_i$ does not appear positively anywhere in $g(\psi)\cup\{\, T(i,I(x_i)) \mid x_i \in \dom(I)\, \}$.
		Accordingly, and by downward closure, $E(i)\cup g(\psi)\cup\{\, T(i,I(x_i)) \mid x_i \in \dom(I)\, \}\models (\F q_i)\lor f(\psi)$ if and only if 
		\begin{equation}
			\exists b\in\{0,1\}:\; T(i,1-b)\models \F q_i \text{ and }
			(E(i)\cup g(\psi)\cup\{\, T(i,I(x_i)) \mid x_i \in \dom(I)\, \}) \setminus T(i,1-b) \models f(\psi).\label{eq:IS-exists-1}
		\end{equation}
		Since $(E(i)\cup g(\psi)\cup\{\, T(i,I(x_i)) \mid x_i \in \dom(I)\, \}) \setminus T(i,1-b) = T(i,b)\cup g(\psi)\cup\{\, T(i,I(x_i)) \mid x_i \in \dom(I)\, \} = g(I[x_i \mapsto b],\psi)$, Equation~\eqref{eq:IS-exists-1} holds if and only if $g(I[x_i \mapsto b],\psi) \models f(\psi)$, for some bit $b\in\{0,1\}$.
		 By the induction hypothesis, the latter holds if and only if there exists a bit $b\in\{0,1\}$ s.t.\ $I[x_i\mapsto b]\models\psi$. 
		 Finally by the semantics of $\exists$  this holds if and only if $I\models\exists x_i\psi$.
		
``Case $\varphi'=\forall x_i\psi$.'' 
		We need to show that $I\models \forall x_i\psi$ if and only if $g(I,\forall x_i\psi)\models f(\forall x_i\psi)$. 
		
		First note that, by the definitions of $f$ and $g$, we have 
		$$g(I,\forall x_i\psi)\models f(\forall x_i\psi)$$ if and only if
\begin{equation}\label{unicase1}
U(i)\cup E(i)\cup g(\psi)\cup\{\, T(i,I(x_i)) \mid x_i \in \dom(I)\, \} \models \bigl(\$\! \lor\! (\lnot q_i\U q_i)\!\lor\! \F[\#\!\land\! \X f(\psi)]\bigr)\U \#.
\end{equation}
In the following, we will show that \eqref{unicase1} is true if and only if $T(i,b) \cup  g(\psi)\cup\{\, T(i,I(x_i)) \mid x_i \in \dom(I)\, \} \models f(\psi)$ for all $b\in\{0,1\}$. 
From this the correctness follows analogously as in the case for the existential quantifier.

Notice first that each trace in $U(i)\cup E(i)\cup g(\psi)\cup\{\, T(i,I(x_i)) \mid x_i \in \dom(I)\, \}$ is periodic with period length either $3$ or $6$, and exactly the last element of each period is marked by the symbol~$\#$. 
Consequently, it is easy to see that \eqref{unicase1} is true if and only if
\begin{equation}\label{unicase2}
(U(i)\!\cup\! E(i)\!\cup\! g(\psi)\cup\{\, T(i,I(x_i)) \mid x_i \in \dom(I)\, \})[j,\infty) \models \$\! \lor\! (\lnot q_i\U q_i)\!\lor\! \F[\#\!\land\! \X f(\psi)],
\end{equation}	
for each $j\in\{0,1,2,3,4\}$.
Note that
$$(U(i)\cup E(i)\cup g(\psi)\cup\{\, T(i,I(x_i)) \mid x_i \in \dom(I)\, \})[j,\infty) \models \$,$$ for each $j\in\{1,2,4\}$, whereas no non-empty subteam of $(U(i)\cup E(i)\cup g(\psi)\cup\{\, T(i,I(x_i)) \mid x_i \in \dom(I)\, \})[j,\infty)$, $j\in\{0,3\}$, satisfies $\$$. Accordingly, \eqref{unicase2} is true if and only if
\begin{equation}\label{unicase3}
(U(i)\cup E(i)\cup g(\psi)\cup\{\, T(i,I(x_i)) \mid x_i \in \dom(I)\, \})[j,\infty) \models (\lnot q_i\U q_i)\!\lor\! \F[\#\land \X f(\psi)],
\end{equation}
for both $j\in\{0,3\}$.
Note that, by construction, $q_i$ does not occur positively in $g(\psi)$. 
As a result, $X\cap g(\psi)[j,\infty)= \emptyset$, $j\in\{0,3\}$, for all teams $X$ s.t. $X\models \lnot q_i\U q_i$.  
Also, none of the symbols $x_{i'}$, $c_{i'}$, $q_{i''}$, for $i'$, $i'' \in \nats$ with $i''\neq i$, occurs positively in $U(i)$. 
On that account,  $X\cap U(i)[j,\infty)= \emptyset$, $j\in\{0,3\}$, for all $X$ s.t.\ $X\models \F[\#\land \X f(\psi)]$, for eventually each trace in $X$ will end up in a team that satisfies one of the formulae of the form $\F x_{i'}$, $\F c_{i'}$, or $\F q_{i''}$ (see the inductive definition of $f$).
Furthermore, notice that $(X \cap\{\, T(k,I(x_k)) \mid x_k \in \dom(I)\, \})[j,\infty) = \emptyset$ for $j \in \{ 0, 3 \}$, for all teams $X$ such that $X \models \neg q_i U q_i$.

Moreover, it is easy to check that $(T(i,1)\cup U(i))[0,\infty)\models \lnot q_i\U q_i$,  $(T(i,0)\cup U(i))[0,\infty)\not\models \lnot q_i\U q_i$, $(T(i,0)\cup U(i))[3,\infty)\models \lnot q_i\U q_i$, and $(T(i,1)\cup U(i))[3,\infty)\not\models \lnot q_i\U q_i$. 
From these, together with downward closure, it follows that \eqref{unicase3} is true if and only if for $b_0=1 $ and $b_3=0$
\begin{equation}\label{unicase4}
(U(i)\cup T(i,b_j))[j,\infty) \models \lnot q_i\U q_i, \text{ for all $j\in\{0,3\}$}
\end{equation}	
and
\begin{equation}\label{unicase5}
(T(i,1-b_j) \cup g(\psi)\cup\{\, T(i,I(x_i)) \mid x_i \in \dom(I)\, \})[j,\infty) \models \F[\#\land \X f(\psi)], 
\end{equation}
for both $j\in\{0,3\}$.
In fact, as \eqref{unicase4} always is the case, \eqref{unicase3} is equivalent with \eqref{unicase5}. 
By construction, \eqref{unicase5} is true if and only if $(T(i,b) \cup g(\psi)\cup\{\, T(i,I(x_i)) \mid x_i \in \dom(I)\, \})[6,\infty) \models f(\psi)$, for both $b\in\{0,1\}$. 
Now, since $$(T(i,b) \cup g(\psi)\cup\{\, T(i,I(x_i)) \mid x_i \in \dom(I)\, \})[6,\infty)=T(i,b) \cup g(\psi)\cup\{\, T(i,I(x_i)) \mid x_i \in \dom(I)\, \}$$ the claim applies.
\end{proof}

Now we turn our attention to proving a matching upper bound. 
To this end, we need to introduce some notation to manipulate team encodings. 
Given a pair~$(t_0,t_1)$ of finite traces~$t_0 = t_0(0) \cdots t_0(n)$ and $t_1 = t_1(0)\cdots t_1(n')$, we define $(t_0,t_1)[1,\infty)$ to be $(t_0(1) \cdots t_0(n), t_1)$ if $t_0 \neq \epsilon$, and to be $(\epsilon, t_1(1) \cdots t_1(n')t_1(0))$ if $t_0 = \epsilon$. 
Furthermore, we inductively define $(t_0,t_1)[i,\infty)$ to be $(t_0,t_1)$ if $i = 0$, and to be $((t_0,t_1)[1,\infty))[i-1,\infty)$ if $i>0$. 
Then, $$\eval{(t_0,t_1)[i,\infty)} = (\eval{(t_0,t_1)})[i,\infty),$$ that is, we have implemented the prefix-removal operation on the finite representation. Furthermore, we lift this operation to team encodings~$\teamup$ by defining $\teamup[i,\infty) = \set{\, (t_0,t_1)[i,\infty) \mid (t_0,t_1) \in \teamup\,}$. As a result, we have $\eval{\teamup[i,\infty)} = (\eval{\teamup})[i,\infty)$.
 
Given a finite team encoding $\teamup $, let $$\prfx(\teamup) = \max \set{\,\size{t_0} \mid (t_0, t_1) \in \teamup\,}$$ and let $\lcm(\teamup)$ be the \emph{least common multiple} of $\set{\,\size{t_1} \mid (t_0, t_1) \in \teamup\,}$. 
Then, $\teamup[i,\infty) = \teamup[i+\lcm(\teamup),\infty)$ for every $i \ge \prfx(\teamup)$. 

Furthermore, observe that if $\teamup$ is a finite team encoding and $i \ge \prfx(\teamup)$, then $\teamup[i,\infty)$ and $\teamup[i+\lcm(\teamup),\infty)$ satisfy exactly the same \TeamLTL-formulae. 

\begin{remark}
\label{rem_pcprops}	
In particular, we obtain the following consequences for temporal operators (for finite $\teamup$):

\begin{tabbing}
    $\teamup \models \psi \U \phi$ \= iff \= Rechts \kill
	$\teamup \models \psi \U \phi$ \> iff \> $\exists k \leq \prfx(\teamup)+\lcm(\teamup)$ such that $ \teamup[k, \infty) \models \phi$ and 
	$\forall k' < k: \teamup[k', \infty) \models \psi$.\\
	$\teamup \models \psi \R \phi$ \> iff \> $\forall k \leq \prfx(\teamup)+\lcm(\teamup)$ we have that $ \teamup[k, \infty) \models \phi$ or 
	$\exists k' < k: \teamup[k', \infty) \models \psi$.
\end{tabbing}
\end{remark}

Accordingly, we can restrict the range of the temporal operators when model checking a finite team encoding. This implies that a straightforward recursive algorithm implementing team semantics solves $\TPC(\TeamLTL)$.

\begin{lemma}\label{lem:TPCs_in_PSPACE}
	$\TPC(\TeamLTL)$ is in $\PSPACE$.
\end{lemma}
\begin{proof}
Consider Algorithm~\ref{algo_pathchecking} where $\varovee$ and $\bigvarovee$ denote classical (meta level) disjunctions, not splitjunctions, which are used to combine results from recursive calls.

\SetKw{proce}{Procedure}
\SetKwFunction{pcheck}{chk}
\SetKwFunction{nitram}{nitram}

\LinesNumbered
\begin{algorithm}\caption{Algorithm for $\TPC(\TeamLTL)$.}\label{algo_pathchecking}
\small
\proce \pcheck{Finite team encoding $\teamup$, formula $\phi$}\;
	\lIf{$\varphi=p$}{\Return $ \bigwedge_{(t_0,t_1)\in\teamup}\,\, p \in t_0t_1(0)$}
	\lIf{$\varphi=\neg p$}{\Return $ \bigwedge_{(t_0,t_1)\in\teamup}\,\, p \notin t_0t_1(0)$}
	\lIf{$\varphi=\psi\land\psi'$}{\Return \pcheck{$\teamup,\psi$} $\land$ \pcheck{$\teamup,\psi'$}}
	\lIf{$\varphi=\psi\lor\psi'$}{\Return $\bigvarovee_{\teamup' \subseteq \teamup}$\,\,\pcheck{$\teamup',\psi$} $\land$ \pcheck{$\teamup\setminus \teamup',\psi'$}}
	\lIf{$\varphi=\X\psi$}{\Return \pcheck{$\teamup[1,\infty),\psi$}}
	\lIf{$\varphi=\psi\U\psi'$}{\Return $\left(\right.$ $\bigvarovee_{k \le \prfx(\teamup) + \lcm(\teamup)}$\,\,\pcheck{$\teamup[k,\infty),\psi'$} $\land$ $\bigwedge_{k' < k}$\,\,\pcheck{$\teamup[k',\infty),\psi$}$\left.\right)$}
	\lIf{$\varphi=\psi\R\psi'$}{\Return $\left(\right.$ $\bigwedge_{k \le \prfx(\teamup) + \lcm(\teamup)}$\,\,\pcheck{$\teamup[k,\infty),\psi'$} $\varovee$ $\bigvarovee_{k' < k}$\,\,\pcheck{$\teamup[k',\infty),\psi$}$\left.\right)$}

\end{algorithm}

The algorithm is an implementation of the  semantics of \TeamLTL with slight restrictions to obtain the desired complexity. In line~5, we only consider strict splits, i.e., the team is split into two disjoint parts. This is sufficient due to downward closure. Furthermore, the scope of the temporal operators in lines~7 and 8 is restricted to the interval $[0, \prfx(\teamup) + \lcm(\teamup)]$. This is sufficient due to Remark~\ref{rem_pcprops}.

It remains to analyse the algorithm's space complexity. Its recursion depth is bounded by the size of the formula. Further, in each recursive call, a team encoding has to be stored. Additionally, in lines 5, 7, and 8, a disjunction or conjunction of  arity~$2^{p(\teamsize{\teamup})}$ for some fixed polynomial~$p$ has to be evaluated. In each case, this only requires polynomial space in the input to make the recursive calls and to aggregate the return value. 
Thus, Algorithm~\ref{algo_pathchecking} is implementable in polynomial space. 
\end{proof}

Combining Lemmas~\ref{lem:TPCs_PSPACEhard} and \ref{lem:TPCs_in_PSPACE} settles the complexity of $\TPC(\TeamLTL)$.
\begin{theorem}\label{thm:tpcs}
	$\TPC(\TeamLTL)$ is $\PSPACE$-complete w.r.t.\ $\leqpm$-reductions.
\end{theorem}

\subsection{Model Checking}
\label{subsec:modelchecking}

Here, we consider the model checking problem $\TMC(\TeamLTL)$. 
We show that model checking for splitjunction-free formulae is in $\PSPACE$ (as is \LTL model checking under standard semantics). Then, we discuss the challenges one has to overcome to generalize this result to formulae with splitjunctions, which we leave as an open problem.

\begin{theorem}\label{thm:tmcs-splitfree}
$\TMC(\TeamLTL)$ restricted to splitjunction-free formulae is in $\PSPACE$.
\end{theorem}
\begin{proof}
Fix a Kripke structure~$\kripke = (W, R, \eta, w_I)$ and a splitjunction-free formula~$\phi$. We define $S_0 = \set{w_I}$ and $S_{i+1} = \set{\,w' \in W \mid (w,w') \in R \text{ for some }w \in S_i\,}$ for all $i \ge 0$. By the pigeonhole principle, this sequence is ultimately periodic with a characteristic~$(s,p)$ satisfying $s+p \le 2^{\size{W}}$.\footnote{The characteristic of an encoding~$(t_0, t_1)$ of an ultimately periodic trace~$t_0 t_1t_1 t_1 \cdots$ is the pair~$(\size{t_0}, \size{t_1})$. Slightly abusively, we say that $(\size{t_0}, \size{t_1})$ is the characteristic of $t_0 t_1t_1 t_1 \cdots$, although this is not unique.} Next, we define a trace~$t$ over $\ap \cup \set{\,\overline{p} \mid p \in \ap\,}$ via
\[
t(i) =\, \set{\,p \in \ap \mid p \in \eta(w) \text{ for all } w \in S_i\,} \cup
\set{\,\overline{p} \mid p \notin \eta(w) \text{ for all } w \in S_i\,}	
\]
that reflects the team semantics of (negated) atomic formulae, which have to hold in every element of the team. 

An induction over the structure of $\phi$ shows that $T(\kripke) \models \phi$ if and only if $t \ltlmodels \overline{\phi}$ (i.e., under classical \LTL semantics), where $\overline{\phi}$ is obtained from $\phi$ by replacing each negated atomic proposition~$\neg p$ by $\overline{p}$. 
To conclude the proof, we show that $t \ltlmodels \overline{\phi}$ can be checked in non-deterministic polynomial space, exploiting the fact that $t$ is ultimately periodic and of the same characteristic as $S_0S_1S_2 \cdots$. 
However, as $s+p$ might be exponential, we cannot just construct a finite representation of $t$ of characteristic~$(s, p)$ and then check satisfaction in polynomial space.

Instead, we present an on-the-fly approach which is inspired by similar algorithms in the literature. It is based on two properties:
\begin{enumerate}
	\item Every $S_i$ can be represented in polynomial space, and from $S_i$ one can compute $S_{i+1}$ in polynomial time. 
	\item For every \LTL-formula~${\overline{\phi}}$, there is an equivalent non-deter\-mi\-nis\-tic B\"uchi automaton~$\aut_{\overline{\phi}}$ of exponential size (see, e.g.,~\cite{BaierKatoen08} for a formal definition of B\"uchi automata and for the construction of $\aut_{\overline{\phi}}$). States of $\aut_{\overline{\phi}}$ can be represented in polynomial space and given two states, one can check in polynomial time, whether one is a successor of the other. 
\end{enumerate}
These properties allow us to construct both $t$ and a run of $\aut_{\overline{\phi}}$ on $t$ on the fly.

In detail, the algorithm works as follows. It guesses a set~$S^* \subseteq W$ and a state~$q^*$ of $\aut_{\overline{\phi}}$ and checks whether there are $i < j$ satisfying the following properties:
\begin{itemize}
	\item $S^* = S_i = S_j$,
	\item $q^*$ is reachable from the initial state of $\aut_{\overline{\phi}}$ by some run on the prefix~$t(0) \cdots t(i)$, and
	\item $q^*$ is reachable from $q^*$ by some run on the infix~$t(i+1) \cdots t(j)$. This run has to visit at least one accepting state. 
\end{itemize}
By an application of the pigeonhole principle, we can assume w.l.o.g.\ that $j$ is at most exponential in $\size{W}$ and in $\size{\phi}$.

Let us argue that these properties can be checked in non-deter\-mi\-nis\-tic polynomial space. Given some guessed $S^*$, we can check the existence of $i<j$ as required by computing the sequence~$S_0 S_1 S_2\cdots$ on-the-fly, i.e., by just keeping the current set in memory, comparing it to $S^*$, then computing its successor, and then discarding the current set. While checking these reachability properties, the algorithm also guesses corresponding runs as required in the second and third property. As argued above, both tasks can be implemented in non-deterministic space. To ensure termination, we stop this search when the exponential upper bound on $j$ is reached. This is possible using a counter with polynomially many bits and does not compromise completeness, as argued above. 

It remains to argue that the algorithm is correct. First, assume $t \ltlmodels \overline{\phi}$, which implies that $\aut_{\overline{\phi}}$ has an accepting run on $t$. Recall that $t$ is ultimately periodic with characteristic~$(s,p)$ such that $s+p \le 2^{\size{W}}$ and that $\aut_{\overline{\phi}}$ is of exponential size. Thus, an application of the pigeonhole principle yields~$i <j$ with the desired properties.

Secondly, assume the algorithm finds~$i<j$ with the desired properties. Then, the run to $q$ and the one from $q$ to $q$ can be turned into an accepting run of $\aut_{\overline{\phi}}$ on $t$. That being so, $t \ltlmodels{\overline{\phi}}$.
\end{proof}

Note that as long as we disallow splitjunctions our algorithm is able to deal with contradictory negations and other extensions; we will return to this topic shortly in the next section.

The complexity of the general model checking problem is left open. It is trivially $\PSPACE$-hard, due to Theorem~\ref{thm:tpcs} and the fact that finite teams of ultimately periodic traces can be represented by Kripke structures.
However, the problem is potentially much harder as one has to deal with infinitely many splits of possibly uncountable teams with non-periodic traces, if a split occurs under the scope of an globally-operator.

\section{Extensions of TeamLTL}\label{sec:extensions}
Next, we take a brief look into extensions of our logic. Extensions present a flexible way to delineate the expressivity and complexity of team-based logics. The philosophy behind extensions is to consider what are the fundamental hyperproperties that we want a logic to be able to express and add those as new atomic expressions.

The most well studied atomic expressions considered in team semantics are dependence and inclusion atoms. Intuitively, the dependence atom $\dep(p_1,\dots,p_n;q_1,\dots,q_m)$ expresses that the truth values of the variables $q_1,\dots,q_m$ are functionally determined by the truth values of $p_1,\dots,p_n$.
Formally, for Teams~$T \subseteq (\pow{\ap})^\omega$, the satisfaction of a dependence atom $T\models\!\dep(p_1,\dots,p_n;q_1,\dots,q_m)$ has the following meaning:
\[
	\forall t,t'\in T:\; (\agreeson{p_1}{t(0)}{t'(0)}\land \dots\land \agreeson{p_n}{t(0)}{t'(0)}) \text{ implies } (\agreeson{q_1}{t(0)}{t'(0)}\land\dots\land\agreeson{q_m}{t(0)}{t'(0)}),
\]
where $\agreeson{p}{t(0)}{t'(0)}$ means the sets $t(0)$ and $t'(0)$ agree on proposition $p$, i.e., both contain $p$ or not.
Observe that the formula $\dep(;p)$ merely means that $p$ has to be constant on the team.
Often, due to convenience we will write $\dep(p)$ instead of $\dep(;p)$. 
Note that the hyperproperty `input determinism', i.e., feeding the same input to a system multiple times yields identical behaviour, now can be very easily expressed via the formula\label{pg:input-nonint}
$
\dep(i_1,\dots,i_n;o_1,\dots,o_m),
$
where $i_j$ are the (public) input variables and $o_j$ are the (public) output variables.

Already in \cite{ehmmvv13} it was realised that in the team semantics of modal logic, it makes sense to allow modal formulae as parameters of dependence atoms. In fact, in the modal logic setting, this is necessary in order to obtain a logic that is expressively complete for downward closed team properties \cite[Corollary 4.5]{helusavi14}. Here, we take an analogous approach and allow parameters to dependence atoms to be LTL-formulae. Intuitively, the dependence atom $\dep(\varphi_1,\dots,\varphi_n;\psi_1,\dots,\psi_m)$ expresses that the truth values of the LTL-formulae $\psi_1,\dots,\psi_m$ are functionally determined by the truth values of $\varphi_1,\dots,\varphi_n$.
Formally, for Teams~$T \subseteq (\pow{\ap})^\omega$, the satisfaction of a dependence atom $T\models\!\dep(\varphi_1,\dots,\varphi_n;\psi_1,\dots,\psi_m)$ has the following meaning:
\[
	\forall t,t'\in T:\; (\agreeson{\varphi_1}{t}{t'}\land \dots\land \agreeson{\varphi_n}{t}{t'}) \text{ implies } (\agreeson{\psi_1}{t}{t'}\land\dots\land\agreeson{\psi_m}{t}{t'}),
\]
where $\agreeson{\varphi}{t}{t'}$ means the traces $t$ and $t'$ agree on the truth value of $\varphi$ with respect to the classical (trace-based) semantics of LTL.

\emph{Inclusion atoms} $(\varphi_1,\dots,\varphi_n) \subseteq (\psi_1,\dots,\psi_n)$ on the other hand express the inclusion dependency that all the truth values occurring for $\varphi_1,\dots,\varphi_n$ must also occur as truth values for $\psi_1,\dots,\psi_n$. Hence, if $o_1,\dots, o_n$ denote public observable bits and $c$ is a bit revealing confidential information, 
then the atom $(o_1,\dots o_n, c) \subseteq (o_1,\dots o_n, \neg c)$ expresses a form of non-inference by stating that an observer cannot infer the value of the confidential bit from the public observables.

One can also take a more abstract view and consider atoms whose semantics can be written in first-order (FO) logic over some trace properties; this includes both dependence and inclusion atoms as special cases.
The notion of generalised atoms in the setting of first-order team semantics was introduced by Kuusisto~\cite{kuusisto15}. 
An $n$-ary generalised atom is an expression of the form $\#(\phi_1,\dots,\phi_n)$ that takes $n$ \LTL-formulae as parameters.
We consider FO-formulae over the signature $(A_{x_i})_{1\leq i \leq n}$, where each $A_{x_i}$ is a unary predicate,  as defining formulae for $n$-ary generalised atoms.
We interpret a team~$T$ as a relational structure~$\mathfrak{A}(T)$ over the same signature with universe~$T$.
When we evaluate $\#(\phi_1,\dots,\phi_n)$, the interpretations of $A_{x_i}$ in $\mathfrak{A}(T)$ are determined by the interpretation of the parameters of $\#$. That is, $t \in T$ is in $ A^{\mathfrak A}_{x_i}$  if and only if $t\models \phi_i$ in the classical semantics of \LTL. 
\begin{definition}
	An FO-formula~$\psi$ \emph{defines} the $n$-ary generalised atom ~$\#$ if $T\models \#(\phi_1,\dots,\phi_n) \Longleftrightarrow \mathfrak A(T)\models\psi$. 
	In this case, $\#$ is also called an \emph{FO-definable} \emph{generalised atom}.\label{def:FO-definable}	
\end{definition}

In the view of FO-definable atoms, the dependence atom $\dep(p;q)$ is FO-definable by 
\[
\forall t\forall t'
\bigl((A_{x_1}(t)\leftrightarrow A_{x_1}(t'))\to(A_{x_2}(t)\leftrightarrow A_{x_2}(t'))\bigr)\]
We call a \TeamLTL-formula extended by a generalised atom~$\#$ a $\TeamLTL(\#)$-formula.
Similarly, we lift this notion and the corresponding decision problems to \emph{sets of generalised atoms} $\mathcal{D}$, i.e., $\TPC(\TeamLTL(\mathcal{D}))$ is the path checking problem for \TeamLTL-formulae which may use the generalised atoms in $\mathcal{D}$. 

Another way to extend \TeamLTL is to introduce additional connectives. Here the usual connectives to consider are the Boolean disjunction $\varovee$ and the contradictory negation $\sim$.
A team $T$ satisfies $\varphi\varovee\psi$ if it satisfies $\varphi$ or $\psi$ (or both).
Contradictory negation combined with team semantics allows for powerful constructions.
For instance, the complexity of model checking for propositional logic jumps from $\NC{1}$ to $\PSPACE$~\cite{muellerDiss}, whereas the complexity of validity and satisfiability jumps all the way to alternating exponential time with polynomially many alternations ($\ATIME(\exp,\pol)$)~\cite{Hannula:2018:CPL:3176362.3157054}.
Formally, we define that $T\models{\sim}\varphi$ if $T\not\models\varphi$. 
Note that the contradictory negation $\sim$ is not equivalent to the negation $\neg$ of atomic propositions defined earlier, i.e., ${\sim}p$ and $\neg p$ are not equivalent.

It turns out that the algorithm for $\TPC(\TeamLTL)$ (Algorithm~\ref{algo_pathchecking} on page~\pageref{algo_pathchecking}) is very robust to strengthenings of the logic via the aforementioned constructs. 
The result of Theorem \ref{thm:tpcs} can be extended to facilitate also the Boolean disjunction, contradictory negation and first-order definable generalised atoms.	

\begin{theorem}\label{thm:GenAtoms} 
	Let $\mathcal{D}$ be a finite set of first-order definable generalised atoms.
	$\TPC(\TeamLTL(\varovee, \sim,\mathcal D))$ is $\PSPACE$-complete w.r.t.\ $\leqpm$-reductions.
\end{theorem}
\begin{proof}
The lower bound applies from Theorem~\ref{thm:tpcs}. 
For the upper bound, we extend the algorithm stated in the proof of Lemma~\ref{lem:TPCs_in_PSPACE} for the cases for the Boolean disjunction, contradictory negation, and FO-definable atoms.
The case for the Boolean disjunction is obtained by adding the following line to the recursive algorithm of the proof of Lemma~\ref{lem:TPCs_in_PSPACE} (where the latter $\varovee$ denotes the classical meta level disjunction):
  	\[\textbf{if } \varphi=\psi\lor\psi' \textbf{ then return } \texttt{chk} (\teamup,\psi) \varovee \texttt{chk}(\teamup,\psi')\]
The case for the contradictory negation is obtained by adding the following line to the recursive algorithm (where $\neg$ denotes the classical meta level negation):
  	\[ \textbf{if } \phi = \sim\!\phi' \textbf{then return } \neg\texttt{chk}(\teamup,\phi') \]
Note that, since the extension is not anymore a downward closed logic, we need to modify the case for the splitjunction to reflect this. Hence we use the case:
\[	
	{\textbf{if } \varphi=\psi\lor\psi' \textbf{ then return } \bigvarovee_{\teamup_1 \cup \teamup_2 = \teamup} \texttt{chk}(\teamup_1,\psi) \land \texttt{chk}(\teamup_2,\psi')}
	\]
Finally, whenever a first-order definable atom $\#$ appears in the computation of the algorithm, we need to solve an FO model checking problem and classical \LTL path checking problems.
As FO model checking is solvable in logarithmic space for any fixed formula~\cite{Immerman1998} and \LTL path checking can be done in NC (combined complexity)~\cite{kf09} the theorem follows.
\end{proof}

The next proposition translates a result from Hannula~et~al.~\cite{Hannula:2018:CPL:3176362.3157054} to our setting. 
They show completeness for $\ATIME(\exp,\pol)$ for the satisfiability problem of propositional team logic with contradictory negation.
This logic coincides with \TeamLTL-formulae without temporal operators. 

\begin{proposition}[\cite{Hannula:2018:CPL:3176362.3157054}]\label{prop:atime(exp,pol)}
	$\TSAT(\TeamLTL(\sim))$ for formulae without temporal operators is complete for the class $\ATIME(\exp,\pol)$ w.r.t.\ $\leqpm$-reductions.
\end{proposition}
The result from the previous proposition will be utilised in the proof of the next theorem.
It shows that already a very simple fragment of $\TeamLTL(\sim)$ has an $\ATIME(\exp,\pol)$-hard model checking problem. L\"uck has established that the general problem is complete for third-order arithmetic~\cite{DBLP:journals/tcs/Luck20}.

\begin{theorem}\label{thm:tmca-tmcs-negation}
$\TMC(\TeamLTL(\sim))$ is $\ATIME(\exp,\pol)$-hard w.r.t.\ $\leqpm$-reductions. 
\end{theorem}
\begin{proof}
We will state a reduction from the satisfiability problem of propositional team logic with the contradictory negation $\sim$ (short $\PL(\sim)$). The stated hardness then follows from Proposition~\ref{prop:atime(exp,pol)}.

\begin{figure}
	\Description[A Kripke structure]{A Kripke structure that starts in the world r and then branches to two successor worlds. Intuitively, there are n column layers of two worlds which are connected to the successor layer in the way that every world is connected to every other in the next layer. The top worlds are called bi and the bottom worlds are called ai. The bi worlds have labeled propositions bar pi and the ai worlds have pi propositions labeled.}
\centering
	\begin{tikzpicture}[x=1.5cm,y=.5cm,thick]
		\node at (-.5,0) {$ \kripke_P$:};
		\node[dw] (root) at (0,0) {$r$};
		\node[d,label={270:$p_1$}] (a1) at (1,-1) {$a_1$};
		\node[d,label={90:$\overline{p_1}$}] (b1) at (1,1) {$b_1$};
		\node[d,label={270:$p_2$}] (a2) at (2,-1) {$a_2$};
		\node[d,label={90:$\overline{p_2}$}] (b2) at (2,1) {$b_2$};
		\node[d,label={270:$p_n$}] (an) at (3,-1) {$a_n$};
		\node[d,label={90:$\overline{p_n}$}] (bn) at (3,1) {$b_n$};
		
		\foreach \f/\t in {root/a1,root/b1,a1/b2,a1/a2,b1/b2,b1/a2}{
			\path[-stealth',draw] (\f) edge (\t);
		}
		
		\path[-stealth',draw,dotted] (a2) edge (bn);
		\path[-stealth',draw,dotted] (a2) edge (an);
		\path[-stealth',draw,dotted] (b2) edge (an);
		\path[-stealth',draw,dotted] (b2) edge (bn);
		
		\path[-, draw] (an) edge[loop right,->, >=stealth'] (an);
		\path[-, draw] (bn) edge[loop right,->, >=stealth'] (bn);

	\end{tikzpicture}
\caption{Kripke structure for the proof of Theorem~\ref{thm:tmca-tmcs-negation}.}\label{fig:kripke-tmca-tmcs-negation}
\end{figure}
For $P=\{p_1,\dots, p_n\}$, consider the tra\-ces starting from the root $r$ of the Kripke structure $\kripke_P$ depicted in Figure~\ref{fig:kripke-tmca-tmcs-negation} using proposition symbols $p_1,\dots, p_n,\overline{p_1},\dots, \overline{p_n}$.
Each trace in the model corresponds to a propositional assignment on $P$.
For $\varphi \in \PL{(\sim)}$, let $\varphi^*$ denote the $\TeamLTL(\sim)$-formula obtained by simultaneously replacing each (non-ne\-gated) variable $p_i$ by $\F p_i$ and each negated variable $\lnot p_i$ by $\F\overline{p_i}$. 
Let $P$ denote the set of variables that occur in $\varphi$. 
Recall that we defined $\top\coloneqq (p \lor \neg p)$ and $\bot\coloneqq p \land \neg p$, then $T(\kripke_P)\models \bigl(\top \lor ((\sim\!\!\bot) \land \varphi^*)\bigr)$ if and only if $T' \models \varphi^*$ for some non-empty $T'\subseteq T(\kripke_P)$. 
It is easy to check that $T'\models\varphi^*$ if and only if the propositional team  corresponding to $T'$ satisfies $\varphi$ and thus the above holds if and only if $\varphi$ is satisfiable.
\end{proof}

In the following, problems of the form $\TSAT(\TeamLTL(\dep))$, etc., refer to the corresponding problem for the extension of $\TeamLTL$ with dependence atoms. 
The following proposition follows from the corresponding result for classical \LTL using downward closure and the fact that on singleton teams dependence atoms are trivially fulfilled.
\begin{proposition}\label{thm:TSATa-TSATs-dep}
$\TSAT(\TeamLTL(\dep))$ is $\PSPACE$-complete w.r.t.\ $\leqpm$-reductions.
\end{proposition}

The following result from Virtema talks about the validity problem of propositional team logic.
\begin{proposition}[\cite{DBLP:journals/iandc/Virtema17}]\label{prop:val-nexptime}
	Validity of propositional logic with dependence atoms is $\NEXPTIME$-complete w.r.t.\ $\leqpm$-reductions.
\end{proposition}

We use this result to obtain a lower bound on the complexity of $\TMC(\TeamLTL(\dep))$.

\begin{theorem}\label{thm:TMCa-TMCs-dep}
$\TMC(\TeamLTL(\dep))$ is $\NEXPTIME$-hard w.r.t.\ $\leqpm$-reductions.
\end{theorem}
\begin{proof}
The proof of this result uses the same construction idea as in the proof of Theorem~\ref{thm:tmca-tmcs-negation}, but this time from a different problem, namely, validity of propositional logic with dependence atoms which settles the lower bound by Proposition~\ref{prop:val-nexptime}.

Let $\varphi$ be a formula of propositional logic with dependence atoms and $P$ be the set of proposition symbols that occur in $\varphi$. Let $\kripke_P$ be the Kripke structure depicted in Figure~\ref{fig:kripke-tmca-tmcs-negation}, which can be clearly constructed from $\varphi$ in polynomial time.
The idea here is that $T(\kripke_P)$ encodes the set all propositional assignments for $P$.
Due to downward closure, the validity of propositional formulae with dependence atoms boils down to model checking the maximal team in the propositional (and not in the trace) setting, which is represented by $T(\kripke_P)$.
Similar to the proof of Theorem~\ref{thm:tmca-tmcs-negation},
let $\varphi^*$ denote the $\TeamLTL(\dep)$-formula obtained from $\varphi$ by simultaneously replacing each (non-ne\-gated) variable $p_i$ by $\F p_i$ and each negated variable $\lnot p_i$ by $\F\overline{p_i}$. 
Now,
$T(\kripke_P)\models \varphi^*$, if and only if, the propositional team consisting of all assignments for $P$ satisfies $\varphi$, if and only if, $\varphi$ is valid (in the sense of Proposition~\ref{prop:val-nexptime}).
\end{proof}

\section{TeamLTL vs.\ HyperLTL} 

 \TeamLTL expresses hyperproperties~\cite{DBLP:journals/jcs/ClarksonS10}, that is, sets of teams, or equivalently, sets of sets of traces. 
\hyltl~\cite{DBLP:conf/post/ClarksonFKMRS14}, which extends \LTL by trace quantification, is another logic expressing hyperproperties. For example, input determinism can be expressed as follows: every pair of traces that coincides on their input variables, also coincides on their output variables (this can be expressed in \TeamLTL by a dependence atom $\dep$ as sketched in Section \ref{sec:extensions}). This results in a powerful formalism that has vastly different properties than \LTL~\cite{DBLP:conf/stacs/Finkbeiner017}. After introducing syntax and semantics of \hyltl, we compare the expressive power of \TeamLTL and \hyltl.

The formulae of \hyltl are given by the grammar 
\[
\phi {} \ddfn  \exists \pi.\phi \mid \forall \pi.\phi \mid \psi, \quad\quad \psi {}  \ddfn {}  p_\pi \mid \neg \psi \mid \psi \lor \psi \mid \X \psi \mid \psi \U \psi,
\]
where $p$ ranges over atomic propositions in $\ap$ and where $\pi$ ranges over a given countable set~$\var$ of \emph{trace variables}. The other Boolean connectives and the temporal operators release~$\R$, eventually~$\F$, and always~$\G$ are derived as usual, due to closure under negation. A sentence is a closed formula, i.e., one without free trace variables.

The semantics of \hyltl is defined with respect to trace assignments, which are partial mappings~$\Pi \colon \var \rightarrow (\pow{\ap})^\omega$. The assignment with empty domain is denoted by $\Pi_\emptyset$. Given a trace assignment~$\Pi$, a trace variable~$\pi$, and a trace~$t$, denote by $\Pi[\pi \rightarrow t]$ the assignment that coincides with $\Pi$ everywhere but at $\pi$, which is mapped to $t$. 
Furthermore, $\suffix{\Pi}{i}$ denotes the assignment mapping every $\pi$ in $\Pi$'s domain to $\Pi(\pi)[i,\infty) $.
For teams~$T$ and trace-assignments~$\Pi$ we define the following:\footnote{Note that we use the same $\models$ symbol for \hyltl that we used for \TeamLTL. It is clear from the context which semantics is used.}
\begin{tabbing}
	$(T, \Pi) \hyltlmodels  \psi_1 \lor \psi_2 $ \= if \= Rechts \kill
	$(T, \Pi) \hyltlmodels p_\pi$\> if \> $p \in \Pi(\pi)(0)$,\\
	$(T, \Pi) \hyltlmodels  \neg \psi$\> if \> $(T, \Pi) \nhyltlmodels  \psi$,\\
	$(T, \Pi) \hyltlmodels  \psi_1 \lor \psi_2 $\> if \> $(T, \Pi) \hyltlmodels  \psi_1$ or $(T, \Pi) \hyltlmodels  \psi_2$,\\
	$(T, \Pi) \hyltlmodels  \X \psi$\> if \> $(T,\suffix{\Pi}{1}) \hyltlmodels  \psi$,\\
	$(T, \Pi) \hyltlmodels  \psi_1 \U \psi_2$\> if \> $\exists k \ge 0$ such that $ (T,\suffix{\Pi}{k}) \hyltlmodels  \psi_2$ and\\
	\>\>$\forall 0 \le k' < k$ we have that $(T,\suffix{\Pi}{k'}) \hyltlmodels  \psi_1$, \\
	$(T, \Pi) \hyltlmodels  \exists \pi.\psi$\> if \> $\exists t \in T$ such that $(T,\Pi[\pi \rightarrow t]) \hyltlmodels  \psi$, and \\
	$(T, \Pi) \hyltlmodels  \forall \pi.\psi$\> if \> $\forall t \in T$ we have that $(T,\Pi[\pi \rightarrow t]) \hyltlmodels  \psi$. 
\end{tabbing}
We say that $T$ satisfies a sentence~$\phi$ if $(T, \Pi_\emptyset) \hyltlmodels  \phi$, and write $T \hyltlmodels  \phi$.

The semantics of \hyltl is synchronous; this can be seen from how the until and next operators are defined. Hence, one could expect that \hyltl is closely related to \TeamLTL as defined here, which is synchronous as well. 
In the following, we refute this intuition.
In fact, \hyltl is closely related to the asynchronous variant of \TeamLTL~\cite{DBLP:conf/foiks/KontinenSV24}.

Formally, a \hyltl sentence~$\phi$ and a \TeamLTL-formula~$\phi'$ are equivalent if for all teams~$T$ we have that $T \hyltlmodels \varphi$ if and only if $T \models \phi'$.
  \begin{theorem}
  \label{theorem_hyltlvsteam}\hfill

  	\begin{enumerate}
  		\item\label{theorem_hyltlvsteam_teamweakerthanshyltl} No \TeamLTL-formula is equivalent to $\exists\pi.p_\pi$.
  		\item\label{theorem_hyltlvsteam_hyltlweakerthansynchrteam} No \hyltl sentence is equivalent to the \TeamLTL formula~$\F p$. 
  	\end{enumerate}
  \end{theorem}

  \begin{proof}
	(\ref{theorem_hyltlvsteam_teamweakerthanshyltl}) 
	Consider $T = \set{\emptyset^\omega, \set{p} \emptyset^\omega}$. 
	We have $T \hyltlmodels  \exists\pi.p_\pi$. 
	Assume there is an equivalent \TeamLTL-formula, call it $\phi$. 
	Then, $T \models \phi$ and thus $\set{\emptyset^\omega} \models \phi$ by downward closure. 
	Hence, by equivalence, $\set{\emptyset^\omega} \hyltlmodels  \exists\pi.p_\pi$, yielding a contradiction.

  (\ref{theorem_hyltlvsteam_hyltlweakerthansynchrteam}) Bozzelli et al.\ proved that the property encoded by the \TeamLTL formula~$\F p$ cannot be expressed in \hyltl~\cite{BozzelliMP15}.
    \end{proof}

Note that these separations are obtained by very simple formulae.
Furthermore, the same separation hold for $\TeamLTL(\dep)$, using the same arguments. 
   
  \begin{corollary}
  \hyltl and \TeamLTL are of incomparable expressiveness.
  \end{corollary}

\section{Team Semantics for CTL}

After having presented team semantics for \LTL, arguably the most important specification language for linear time properties, we now turn our attention to branching time properties.
Here, we present team semantics for \CTL. 

Recall that in the \LTL setting teams are sets of traces, as \LTL is evaluated on a single trace. 
\CTL is evaluated on a vertex of a Kripke structure. 
Consequently, teams in the \CTL setting are (multi-)sets of vertices. 
For the complexity results of this paper, it does not really matter whether team semantics for \CTL is defined with respect to multi-sets or plain sets.
However, we have chosen multi-sets since the subsequent logic will be conceptually simpler than a variant based on sets. 
See the work of \citeauthor{DBLP:conf/mfcs/KontinenSV23}~\cite{DBLP:conf/mfcs/KontinenSV23} for a development of asynchronous \TeamLTL under set-based semantics. The challenges there to define the right set-based semantics for asynchronous \TeamLTL are analogous to what would rise from the use of path quantifiers in \CTL.

We will reuse some notions from \LTL, e.g., the countably infinite set of propositions is $\ap$. 
The set of all \CTL-formulae is defined inductively via the following grammar:
$$
\varphi\Coloneqq
p\mid 
\lnot p\mid 
\varphi\land\varphi\mid
\varphi\lor\varphi\mid
\logicOpFont{P}\X\varphi\mid 
\logicOpFont{P}[\varphi\U\varphi]\mid 
\logicOpFont{P}[\varphi\R\varphi],
$$
where $\logicOpFont{P}\in\{\A,\E\}$ and $p\in\PROP$. 
We define the following usual shorthands: 
$\top\coloneqq p\vee \neg p$, 
$\bot\coloneqq p\wedge \neg p$, 
$\F\varphi \coloneqq [\top\U\varphi]$, and 
$\G \varphi \coloneqq [\bot\R\varphi]$. 
Note that the formulae are in negation normal form (NNF) which is the convention in team semantics. 
In the classical setting this is not a severe restriction as transforming a given formula into its NNF requires linear time in the input length.
In team semantics, the issue is more involved as here $\neg \phi$ does not have an agreed compositional semantics and the contradictory negation often increases the complexity and expressiveness of a logic considerably as was seen in Section \ref{sec:extensions}.

We will again use Kripke structures as previously defined for \LTL. 
When the initial world is not important, we will omit it from the tuple definition. 
By $\Pi(w)$, for a world $w$ in a Kripke structure~$\kripke$, we denote the (possibly infinite) set of all paths $\pi$ for which $\pi(0)=w$.
For a set $V\subseteq W$, we define $\Pi(V)\dfn \bigcup_{w\in V}\Pi(w)$.

\begin{definition}
Let $\kripke=(W,R,\eta)$ be a Kripke structure and let $w\in W$ a world. 
The satisfaction relation $\ctlmodels$ for \CTL is defined as follows:

\begin{tabbing}
$\kripke,w\ctlmodels\logicOpFont{A}[\varphi\U\psi]$  \= if \= $p\in\eta(w$),\kill
 $\kripke,w\ctlmodels p$ \> if \> $p\in\eta(w$),\\
 $\kripke,w\ctlmodels \lnot p$ \> if \> $p\notin\eta(w)$,\\
 $\kripke,w\ctlmodels \varphi\land\psi$ \> if \> $\kripke,w\ctlmodels\varphi$  and $\kripke,w\ctlmodels\psi$,\\
 $\kripke,w\ctlmodels \varphi\lor\psi$ \> if \> $\kripke,w\ctlmodels\varphi$ or $\kripke,w\ctlmodels\psi$,\\
 $\kripke,w\ctlmodels\logicOpFont{E}\X\varphi$ \> if \> $\exists\pi\in\Pi(w)$ such that $\kripke,\pi(1)\ctlmodels\varphi$,\\
 $\kripke,w\ctlmodels\logicOpFont{E}[\varphi\U\psi]$ \> if \> $\exists\pi\in\Pi(w)\exists k\in\N$ s.t.\ $\kripke,\pi(k)\ctlmodels\psi$ and $\forall  k'<k$ we have that $\kripke,\pi(k')\ctlmodels\varphi$,\\
  $\kripke,w\ctlmodels\logicOpFont{E}[\varphi\R\psi]$ \> if \> $\exists\pi\in\Pi(w)\forall k\in\N$ we have that $\kripke,\pi(k)\ctlmodels\psi$ or $\exists  k'<k$ s.t.\ $\kripke,\pi(k')\ctlmodels\varphi$,\\
 $\kripke,w\ctlmodels\logicOpFont{A}\X\varphi$ \> if \> $\forall\pi\in\Pi(w)$ we have that $\kripke,\pi(1)\ctlmodels\varphi$,\\
 $\kripke,w\ctlmodels\logicOpFont{A}[\varphi\U\psi]$ \> if \> $\forall\pi\in\Pi(w)\exists k\in\N$ s.t.\ $\kripke,\pi(k)\ctlmodels\psi$ and $\forall  k'<k$ we have that $\kripke,\pi(k')\ctlmodels\varphi$,\\
 $\kripke,w\ctlmodels\logicOpFont{A}[\varphi\R\psi]$ \> if \> $\forall\pi\in\Pi(w)\forall k\in\N$ we have that $\kripke,\pi(k)\ctlmodels\psi$ or $\exists  k'<k$ s.t.\ $\kripke,\pi(k')\ctlmodels\varphi$.
\end{tabbing}
\end{definition}

Next, we will introduce team semantics for \CTL based on multisets, obtaining the logic \TeamCTL.
Notice that, for \TeamLTL, a team is a set of traces while here it is a multiset of worlds mimicking the behaviour of satisfaction in vanilla \CTL.
A \emph{multiset} is a generalisation of the concept of a set that allows for having multiple instances of the same element. 
Technically, we define that a multiset is a functional set of pairs $(i,x)$, where $i$ is an element of some sufficiently large set $\Ind$ of indices, and $x$ is an element whose multiplicities we encode in the multiset.
The \emph{support} of a multiset $M$ is the image set $M[\Ind]\dfn \{x \mid (i,x)\in M,  i\in\Ind\}$, while the \emph{multiplicity} of an element $x$ in a multiset $M$ is the cardinality of the pre-image $M^{-1}[\{x\}] \dfn\{ i\in \Ind \mid (i,x) \in M \}$.
We often omit the index, and write simply $x$ instead of $(i,x)$ to simplify the notation.
We also write $\mset{\dots}$ to denote a multiset with its elements inside the curly brackets omitting the indices. 
For multisets $A$ and $B$, we write $A\uplus B\dfn \{((i,0),x) \mid (i,x)\in A \} \cup \{((i,1),x) \mid (i,x)\in B \}$ to denote the disjoint union of $A$ and $B$. 
Note that $A\cup B$ denotes the standard union of the sets $A$ and $B$, which is not always a (functional) multiset.
In the following, we introduce some notation for multisets.

\begin{definition}
	Let $A,B$ be two multisets over some index set $\Ind$.
	We write $A\meq B$ if there exists a permutation $\pi\colon\Ind\to\Ind$ such that $A=\{\,(\pi(i),x)\mid (i,x)\in B\,\}$. 
	Likewise, we write $A\msube B$ if there exists a permutation $\pi\colon\Ind\to\Ind$ such that $A\subseteq\{\,(\pi(i),x)\mid (i,x)\in B\,\}$. 
\end{definition}
Now, we turn to the definition of teams in the setting of \CTL. 
Note that this notion is different from the notion of a team in the context of \LTL, where a team is a possibly infinite set of traces.

\begin{definition}
Let $\kripke=(W,R,\eta)$ be a Kripke structure. 
Any multiset $T$ where $T[\Ind]\subseteq W$ is called a \emph{(multi)team of $\kripke$}.
\end{definition}

By abuse of notation, if the indices are irrelevant, we will write $w\in T$ instead of $(i,w)\in T$.

Let $\kripke=(W,R,\eta)$ be a Kripke structure and $T$ be a team of $\kripke$. 
If $f\colon T \to \Pi(W)$ is a function that maps elements of the team to paths in $W$ such that $f((i,x))\in\Pi(x)$ for all $(i,x)\in T$ then we call $f$ \emph{$T$-compatible}, and define $T[f]$ to be the following multiset of traces
\[
T[f]\coloneqq \{\, (i,f((i,x))) \mid (i,x)\in T\, \}.
\]
Let $n\in\nats$.
We define $T[f,n]$ to be the team of $\kripke$
\[
	T[f,n]\coloneqq\{\,(i,t(n))\mid (i,t)\in T[f]\,\},
\]
that is, the multiset of worlds reached by synchronously proceeding to the $n$-th element of each trace. 

The syntax of \TeamCTL is the same as for \CTL above.
Next, we define (synchronous) team semantics for \CTL.

\begin{definition}\label{def:CTLsemantics}
Let $\kripke=(W,R,\eta)$ be a Kripke structure, $T$ be a team of $\kripke$, and $\varphi,\psi$ be \CTL-formulae. 
The  satisfaction relation $\models$ for \TeamCTL is defined as follows.
\begin{tabbing}
$\kripke,T \models\logicOpFont{A}[\varphi\R\psi]$ \= if \=$\forall\, \compat{T}\, f \, \forall k\in\N$ we have that $\kripke,T [f,k]\models\psi$ or \kill
$\kripke,T \models p$ \>  if \> $\forall w\in T$: $p\in\eta(w)$.\\
$\kripke,T \models \lnot p$  \> if \> $\forall w\in T$: $p\notin\eta(w)$.\\
$\kripke,T \models \varphi\land\psi$ \> if \> $\kripke,T \models\varphi$ and $\kripke,T \models\psi$.\\
$\kripke,T \models \varphi\lor\psi$ \> if \> $\exists T_1 \cup T_2 = T$ s.t.\ $T_1 \cap T_2 = \emptyset,\, \kripke,T _1\models\varphi$ and $\kripke,T _2\models\psi$.\\
$\kripke,T \models\logicOpFont{E}\X\varphi$ \> if \> $\exists\, \compat{T}\, f$ s.t.\ $\kripke,T [f,1]\models\varphi$.\\
$\kripke,T \models\logicOpFont{E}[\varphi\U\psi]$ \> if \> $\exists\, \compat{T}\, f \, \exists k\in\N$ s.t.\ $\kripke,T [f,k]\models\psi$ and $\forall k'<k: \kripke, T[f,k']\models\varphi$\\
$\kripke,T \models\logicOpFont{E}[\varphi\R\psi]$ \> if \>$\exists\, \compat{T}\, f \, \forall k\in\N$: $\kripke,T [f,k]\models\psi$ or $\exists k'<k: \kripke, T[f,k']\models\varphi$\\
$\kripke,T \models\logicOpFont{A}\X\varphi$ \> if \> $\forall\, \compat{T}\, f$: $\kripke,T [f,1]\models\varphi$.\\
$\kripke,T \models\logicOpFont{A}[\varphi\U\psi]$ \> if \>$\forall\, \compat{T}\, f \, \exists k\in\N$ s.t.\ $\kripke,T [f,k]\models\psi$ and $\forall k'<k$: $\kripke, T[f,k']\models\varphi$.\\
$\kripke,T \models\logicOpFont{A}[\varphi\R\psi]$ \> if \>$\forall\, \compat{T}\, f \, \forall k\in\N$: $\kripke,T [f,k]\models\psi$ or $\exists k'<k$ s.t.\ $\kripke, T[f,k']\models\varphi$.
\end{tabbing}
\end{definition}

\paragraph*{Basic Properties}\label{par:properties}

In the following, we investigate several fundamental properties of the satisfaction relation. 
It is immediate that for every $\CTL$-formula $\varphi$ and multiteams $T$ and $S$ such that $T\meq S$, it holds that $\kripke,T \models \varphi$ if and only if $\kripke,S \models \varphi$.
Notice that the following properties also hold for \TeamLTL as defined in Section~\ref{sec:TeamLTL} on page~\pageref{fig:structuralpropertiesoverview}.
Observe that $\kripke, T \models \bot$ if and only if $T=\emptyset$. 
The proof of the following proposition is trivial.

\begin{proposition}[Empty team property]
For every Kripke structure~$\kripke$ we have that $\kripke, \emptyset \models \varphi$ holds for every \CTL-formula $\varphi$.
\end{proposition}

When restricted to singleton teams, team semantics coincides with the traditional semantics of \CTL defined via pointed Kripke structures.
\begin{proposition}[Singleton equivalence]\label{prop:singleton}
 For every Kripke structure $\kripke=(W,R,\eta)$, every world $w\in W$, and every \CTL-formula~$\varphi$ the following equivalence holds:
 $$
\kripke,\mset{w}\models\varphi\Leftrightarrow \kripke,w\ctlmodels\varphi.
 $$
\end{proposition}
\begin{proof}
Let $\kripke=(W,R,\eta)$ be an arbitrary Kripke structure. We prove the claim via induction on the structure of $\varphi$:

Assume that $\varphi$ is a (negated) proposition symbol $p$. Now
\begin{align*}
\kripke,w&\ctlmodels \varphi\\
&\text{iff}\quad \text{$p$ is (not) in $\eta(w)$}\\
&\text{iff}\quad \text{for all $w'\in \mset{w}$ it holds that $p$ is (not) in $\eta(w')$}\\
&\text{iff}\quad \kripke,\mset{w}\models\varphi.
\end{align*}

The case $\land$ is trivial. 
For the $\lor$ case, assume that $\varphi = \psi\lor\theta$. Now it holds that
\begin{align*}
\kripke,&w\ctlmodels \psi\lor\theta\\ 
&\text{iff}\quad \kripke,w\ctlmodels \psi \text{ or } \kripke,w\ctlmodels \theta\\
&\text{iff}\quad \kripke,\mset{w}\models \psi \text{ or } \kripke,\mset{w}\models \theta\\
&\text{iff}\quad (\kripke,\mset{w}\models \psi \text{ and } \kripke,\emptyset \models \theta) \text{ or } (\kripke,\emptyset \models \psi \text{ and } \kripke,\mset{w}\models \theta)\\
&\text{iff}\quad \exists T_1\cup T_2 = \mset{w}\text{ s.t.\ } T_1\cap T_2 = \emptyset,\, \kripke,T _1\models\psi \text{ and }\kripke,T _2\models\theta\\
&\text{iff}\quad \kripke,\mset{w}\models \psi\lor\theta.
\end{align*}

Here the first equivalence holds by the semantics of disjunction, the second equivalence follows by the induction hypothesis, the third via the empty team property, the fourth equivalence follows since there are only two ways to partition a singleton to an ordered pair of sets, and the last by the team semantics of disjunction.

The cases for $\E\X$ and $\A\X$, until and release are all similar and straightforward. 
We show here the case for $\E\X$. 
Assume $\varphi = \E\X\psi$. 
Then, $\kripke,w\ctlmodels \E\X\psi$ iff there exists a trace $\pi\in\Pi(w)$ such that $\kripke,\pi(1) \ctlmodels \psi$. 
Thus, by the induction hypothesis $\kripke,\pi(1) \ctlmodels \psi$ iff $\kripke,\mset{\pi(1)} \models \psi$.
Hence, equivalently, there exists a $\mset{w}$-compatible function $f$ such that $\kripke,\mset{w} [f,1]\models\psi$ which is equivalent to $\kripke,\mset{w} \models \EX\psi$.
\end{proof}

\TeamCTL is \emph{downward closed} if the following holds for every Kripke structure $\kripke$, for every \CTL-formula $\varphi$, and for all teams $T$ and $T'$ of $\kripke$:
\[
\text{If $\kripke,T \models \varphi$ and $T'\msube T$ then $\kripke,T '\models \varphi$}. 
\]

The proof of the following lemma is analogous with the corresponding proofs for modal and first-order dependence logic (see~\cite{va07,va08}).
\begin{proposition}[Downward closure]\label{prop:dwclos}
 \TeamCTL is downward closed.
\end{proposition}
\begin{proof}
The proof is by induction on the structure of $\varphi$. 
Let $\kripke=(W,R,\eta)$ be an arbitrary Kripke structure and $T'\msube T$ be some teams of $\kripke$.
Hence, there is a bijection~$\pi$ witnessing $T'\msube T$, i.e., it satisfies $T'\subseteq\{\,(\pi(i),x)\mid (i,x)\in T\,\}$. 

The cases for literals are trivial: 
Assume $\kripke,T \models p$. 
Then by definition $p\in\eta(w)$ for every $w\in T$. 
Now since $T'\msube T$, clearly $p\in\eta(w)$ for every $w\in T'$. 
Thus, we have that $\kripke,T '\models p$. 
The case for negated proposition symbols is identical.
 
The case for $\land$ is clear. 
For the case for $\varphi\lor\psi$ assume that $\kripke,T \models\varphi\lor\psi$. 
Now, by the definition of disjunction there exist $T_1\cup T_2 = T$ such that $T_1\cap T_2 = \emptyset$, $\kripke,T _1\models\varphi$ and $\kripke,T _2\models\psi$. 
Define $T_1'= \{\,(\pi(i),x) \in T' \mid (i,x)\in T_1\,\}$ and $T_2'$ analogously.
Then, we have $T_1'\msube T_1$ and $T_2'\msube T_2$ (the inclusions are both witnessed by $\pi$ from above).
Hence, by induction hypothesis it then follows that $\kripke,T_1'\models\varphi$ and $\kripke,T_2'\models\psi$. 
Finally, since $T' = T_1' \cup T_2'$ and $T_1' \cap T_2'=\emptyset$, it follows by the semantics of the disjunction that $\kripke,T '\models\varphi\lor\psi$.

Now consider $\EX\varphi$ and assume that $\kripke,T \models\EX\varphi$. 
We have to show that $\kripke,T '\models\EX\varphi$. 
Notice that by the semantics of $\EX$ there exists a $T$-compatible function $f$ such that $\kripke,T [f,1]\models\varphi$. 
Now, since $T'\msube T$, the function $f\upharpoonright T'$ (i.e., the restriction of $f$ to the domain of $T'$) is $T'$-compatible. 
Consequently, by induction hypothesis, $\kripke,T '[f\upharpoonright T',1]\models\varphi$ and thus $\kripke,T '\models\EX\varphi$.

The proofs for the cases for all remaining operators are analogous.
\end{proof}

Similar to \TeamLTL, \TeamCTL violates flatness and union closure. 
In this setting, a logic is \emph{union closed} if it satisfies
\[
\text{$\kripke,T \models\varphi$ and $\kripke,T '\models\varphi$ implies $\kripke,T \uplus T'\models\varphi$,}
\]
and has the \emph{flatness} property if
\[
\text{$\kripke,T \models\varphi$ if and only if for all $w\in T$ it holds that $\kripke,\mset{w}\models\varphi$}.
\]
\begin{proposition}
\CTL is neither union closed nor flat.
\end{proposition}
\begin{proof}
The counter example for both is depicted in Figure~\ref{fig:async-vs-synch} (left). 
\end{proof}
In addition, as exemplified in the right-hand side of Figure~\ref{fig:async-vs-synch}, \TeamCTL is sensitive for multiplicities.

\begin{figure}
	\Description[Two examples showing that multiplicites matter.]{The left example has two traces x and y which both are paths of length 3. Trace x has a p labeled in the third world while trace y has in the second world. The right example has one root world w and then it branches into two paths. The first has a p labeled in the second world and the other in the first.}
 \centering
\begin{tikzpicture}[c/.style={circle,fill=black,inner sep=0mm,minimum width=2mm},x=1cm,y=1cm]

\foreach \x/\l in {1/,2/,3/$p$,4/}
 \node[c,label={90:\l}] (a\x) at (\x,1) {};

\foreach \x/\l in {1/,2/$p$,3/,4/}
 \node[c,label={90:\l}] (b\x) at (\x,0) {};

\foreach \f/\t in {1/2,2/3,3/4}{
 \path[-stealth',black] (a\f) edge (a\t);
 \path[-stealth',black] (b\f) edge (b\t);
}

		\path[-stealth',black] (a4) edge[loop right] ();
		\path[-stealth',black] (b4) edge[loop right] ();

\node[label={180:\texttt{x}}] at (a1) {};
\node[label={180:\texttt{y}}] at (b1) {};

\draw[black,dashed,rounded corners] (.5,1.3) rectangle (1.3,-.3);
\node at (1,1.6) {$T$};
\end{tikzpicture}
\quad\quad\quad\quad
\begin{tikzpicture}[c/.style={circle,fill=black,inner sep=0mm,minimum width=2mm},x=1cm,y=1cm]

	\node[c,label={180:\texttt{w}}] (1) at (1,.5) {};

	\foreach \x/\l in {2/,3/$p$,4/}
	 \node[c,label={90:\l}] (a\x) at (\x,1) {};
	
	\foreach \x/\l in {2/$p$,3/,4/}
	 \node[c,label={90:\l}] (b\x) at (\x,0) {};
	
	\foreach \f/\t in {2/3,3/4}{
	 \path[-stealth',black] (a\f) edge (a\t);
	 \path[-stealth',black] (b\f) edge (b\t);
	}

	\path[-stealth',black] (1) edge (a2);
	\path[-stealth',black] (1) edge (b2);

		\path[-stealth',black] (a4) edge[loop right] ();
		\path[-stealth',black] (b4) edge[loop right] ();
	
	\draw[black,opacity=0,rounded corners] (.7,1.3) rectangle (1.3,-.3);
	\draw[black,dashed,rounded corners] (.4,.8) rectangle (1.3,.2);
	\node at (1,1.1) {$S$};
	\end{tikzpicture}
	\caption{(left) A team~$T$ does not satisfy the \TeamCTL formula~$\EF p$ while all singleton teams~$\mset{v} \subseteq T$ individually do satisfy the formula. (right) The multiplicity of elements in a team matters: If the multiplicity of the world~\texttt{w} in $S$ is $1$ then $S$ satisfies $\AF p$, while for larger multiplicities $S$ falsifies $\AF p$.}\label{fig:async-vs-synch}
\end{figure}

\begin{proposition}
Let $T$ and $T'$ be teams with the same support. 
In general, it need not hold that $\kripke,T \models \varphi$ if and only if $\kripke,T '\models \varphi$.
\end{proposition}

\paragraph*{Comparison to TeamCTL with Time Evaluation Functions}
After the publication of the conference version of this work~\cite{kmv15}, \citeauthor{DBLP:conf/lics/GutsfeldMOV22}~\cite{DBLP:conf/lics/GutsfeldMOV22} revisited the foundations of temporal team semantics and introduced several logics based on a concept they aptly named \emph{time evaluation functions} (tef). 
Note first that, using a slightly different notation than in this paper, the satisfying element in (synchronous) \TeamLTL is a pair $(T,i)$, where $T$ is a set of traces and $i\in\mathbb{N}$ is the current time step. In this paper, we write $T[i,\infty]$ instead of $(T,i)$. 
Hence in synchronous \TeamLTL the traces share access to a common clock. 
Simply put, time evaluation functions $\tau\colon \mathbb{N} \times T \to \mathbb{N}$ model asynchronous evaluation of time on traces which are equipped with local clocks. 
The idea is that $\tau(i,t)$ denotes the value of the local clock of trace $t$ at global time $i$. 
Hence the satisfying element in this setting is a triplet~$(T,\tau,n)$, where $T$ is a set of traces, $\tau$ is a time evaluation function, and $n$ is the value of the global clock. 
\TeamLTL with tef's is then defined similar to our \TeamLTL, where logical operators modify $T$ and $n$. 
In particular, if $\tau$ is the synchronous tef (i.e., $\tau(i,t)=i$ for all $i$ and $t$), then exactly synchronous \TeamLTL, as defined in this paper, is obtained. 
Finally they consider a logic they also named \TeamCTL which takes the syntax of \CTL and reinterprets path quantifiers as quantifiers raging over tef's. 
For further details, we refer to the work of \citeauthor{DBLP:conf/lics/GutsfeldMOV22}~\cite{DBLP:conf/lics/GutsfeldMOV22}.

The team semantics of Definition \ref{def:CTLsemantics} for \CTL can be modified to a team semantics for the full branching time logic $\CTL^\star$ by working with multisets of traces with time steps $(T[f],n)$ instead of teams $T[f,n]$. 
This approach is essentially equivalent to synchronous time evaluation functions as in the work of \citeauthor{DBLP:conf/lics/GutsfeldMOV22}~\cite{DBLP:conf/lics/GutsfeldMOV22}, when interpreted over computation trees.

\section{Expressive Power of TeamCTL}\label{sec:expressive}

Next, we study the expressiveness of \TeamCTL and exemplify that some simple properties are not expressible in \CTL with classical semantics.
In order to relate the team semantics to the classical semantics, we first present a definition that lifts the classical semantics to multisets of worlds.

\begin{definition}\label{def:classical-CTL-multiteam-semantics}
	For every Kripke structure $\kripke=(W,R,\eta)$, every \CTL-formula $\varphi$, and every team $T$ of $\kripke$, we let 
	\[
		\kripke,T \ctlmodelsnew\varphi \text{ if and only if } \forall w\in T: \kripke,w\ctlmodels\varphi.
	\]
	We refer to this notion as \emph{classical (multiset) semantics}.
\end{definition}

Observe that classical multiset semantics is flat by definition (i.e., $\kripke,T \ctlmodelsnew\varphi$ if and only if for all $w\in T$ we have that $\kripke,\mset{w}\ctlmodelsnew\varphi$), and hence union and downward closed.
In the following, we introduce the notions of definability and $k$-definability.
\begin{definition}
For each \CTL-formula $\varphi$, define
\begin{align*}
&\mathfrak{F}^c_\varphi\coloneqq\{(\kripke,T )\mid \kripke,T  \ctlmodelsnew\varphi \text{ under classical multiset semantics}\} \text{ and}\\
&\mathfrak{F}^t_\varphi\coloneqq\{(\kripke,T )\mid \kripke,T \models\varphi \text{ under team semantics}\}.
\end{align*}

We say that $\varphi$ \emph{defines the class $\mathfrak{F}^c_\varphi$ in classical multiset semantics}. 
Analogously, we say that $\varphi$ \emph{defines the class $\mathfrak{F}^t_\varphi$ in  team semantics}.

A class $\mathfrak{F}$ of pairs of Kripke structures and teams is \emph{definable in classical multiset semantics} (in team semantics) if there exists some \TeamCTL-formula~$\psi$ such that $\mathfrak{F}=\mathfrak{F}^c_\psi$ ($\mathfrak{F}=\mathfrak{F}^t_\psi$). 
Furthermore, for $k\in\mathbb{N}$, define
\begin{align*}
\mathfrak{F}^{c,k}_\varphi\coloneqq\{(\kripke,T )\mid \kripke,T \ctlmodelsnew\varphi \text{ and } \lvert T\rvert \leq k \} \text{ and } 
\mathfrak{F}^{t,k}_\varphi\coloneqq\{(\kripke,T )\mid \kripke,T \models\varphi \text{ and } \lvert T\rvert \leq k \}.
\end{align*}

We say that $\varphi$ \emph{$k$-defines the class $\mathfrak{F}^{c,k}_\varphi$ ($\mathfrak{F}^{t,k}_\varphi$, resp.) in classical multiset (team, resp.) semantics}. 
The definition of \emph{$k$-definability} is analogous to that of definability. 
\end{definition}

Next we will show that there exists a class $\mathfrak{F}$ which is definable in  team semantics, but is not definable in classical multiset semantics.
\begin{theorem}\label{thm:undef}
The class $\mathfrak F_{\EF p}^t$ is not definable in classical multiset semantics.
\end{theorem}
\begin{proof}

For the sake of a contradiction, assume that there is a \CTL-formula~$\varphi$ such that $\mathfrak F^c_\varphi = \mathfrak F_{\EF p}^t$. 
Consider the Kripke structure from Figure~\ref{fig:async-vs-synch} (left). 
Clearly $\kripke,\mset{\texttt{x}}\models\EF p$ and $\kripke,\mset{\texttt{y}}\models\EF p$. 
Thus, by our assumption, it follows that $\kripke,\mset{\texttt{x}}\ctlmodelsnew\varphi$ and $\kripke,\mset{\texttt{y}}\ctlmodelsnew\varphi$. 
From Definition~\ref{def:classical-CTL-multiteam-semantics} it then follows that $\kripke,\mset{\texttt{x},\texttt{y}}\ctlmodelsnew\varphi$. 
But clearly $\kripke,\mset{\texttt{x},\texttt{y}}\not\models\EF p$.
\end{proof}

\begin{corollary}
For $k>1$, the class $\mathfrak F_{\EF p}^{t,k}$ is not $k$-definable in classical multiset semantics.
\end{corollary}

We conjecture that the class $\mathfrak F_{\EF p}^{c}$ is not definable in team semantics.  
However, we will see that the class $\mathfrak F_{\EF p}^{c,k}$ is $k$-definable in  team semantics.

\begin{theorem}
	For every $k\in\mathbb{N}$ and every \CTL-formula~$\varphi$, the class $\mathfrak F_{\varphi}^{c,k}$ is $k$-definable in team semantics.
\end{theorem}
\begin{proof}
	Fix $k\in \mathbb{N}$ and a \CTL formula~$\varphi$. We define
	\[
	\varphi' \coloneqq \bigvee_{1\leq i \leq k} \varphi.
	\]
	
	We will show that $\mathfrak F_{\varphi}^{c,k} = \mathfrak F_{\varphi'}^{t,k}$. 
	Let $\kripke$ be an arbitrary Kripke structure and $T$ be a team of $\kripke$ of size at most $k$. Then the following is true
	\begin{align*}
	\kripke,T  \ctlmodelsnew \varphi \quad\Leftrightarrow\quad& \forall w\in T : \kripke,w \ctlmodels \varphi \\
	\Leftrightarrow\quad& \forall w\in T : \kripke,\mset{w} \models \varphi \\
	\Leftrightarrow\quad& \kripke,T  \models \varphi'.
	\end{align*}
	
	The first equivalence follows by Definition~\ref{def:classical-CTL-multiteam-semantics}, the second is due to Proposition~\ref{prop:singleton}, and the last by the semantics of disjunction, empty team property, and downward closure.
\end{proof}

\section{Complexity Results for TeamCTL}\label{sec:complexity}

Next, we define the most important decision problems for \TeamCTL and classify their computational complexity.
Notice that, in comparison to \TeamLTL, the problems are defined in a way that a team is explicitly given as part of the input, while for \TeamLTL, the team is given implicitly by a Kripke structure. 

\problemdef{$\TMC(\TeamCTL)$ --- \TeamCTL Model Checking.}{A Kripke structure $\kripke$, a finite team $T$ of $\kripke$, and a \TeamCTL-formula $\varphi$.}{$\kripke,T \models\varphi$?}

\problemdef{$\TSAT(\TeamCTL)$ --- \TeamCTL Satisfiability.}{A \TeamCTL-formula $\varphi$.}{Does there exist a Kripke structure $\kripke$ and a non-empty team $T$ of $\kripke$ such that $\kripke,T \models\varphi$?}

\subsection{Model Checking}

We first investigate the computational complexity of model checking. Our benchmark here is the complexity of model checking for classical \CTL.

\begin{proposition}[\cite{clemsi86,sc02}]\label{prop:MC-P}
 Model checking for \CTL-formulae under classical semantics is $\P$-complete.
\end{proposition}

Now we turn to the model checking problem for \TeamCTL. 
Here we show that the problem becomes intractable under reasonable complexity class separation assumptions, i.e., $\P\neq\PSPACE$. 
The main idea is to exploit the synchronicity of the team semantics in a way to  check in parallel all clauses of a given quantified Boolean formula for satisfiability by a set of relevant assignments.

\newcommand{\lits}{\textrm{literals}}
\begin{lemma}\label{lem:mcs-pspacehard}
	$\TMC(\TeamCTL)$ is $\PSPACE$-hard  w.r.t.\ $\leqpm$-reductions.
\end{lemma}
\begin{proof}
\renewcommand{\Game}{Q}
We will reduce from $\qbfval$. 
Let $\varphi\dfn\exists x_{1}\forall x_{2}\cdots\Game x_{n}\bigwedge_{j=1}^{m}\bigvee_{i=1}^{3}\ell_{j,i}$ be a closed quantified Boolean formula ($\QBF$) and $\Game=\exists$ if $n$ is odd, resp., $\Game=\forall$ if $n$ is even. 

Now define the corresponding structure $\kripke\coloneqq(W,R,\eta)$ over the set $\PROP = \{x_1, ..., x_n\}$ of atomic propositions as follows (the structure is given in Figure~\ref{fig:generalview-pspace-struc}). 
The structure $\kripke$ is constructed out of smaller structures in a modular way. 
Formally, if $\kripke_1=(W_1,R_1,\eta_1)$ and $\kripke_2=(W_2,R_2,\eta_2)$ are two Kripke structures, then let $\kripke_1\cup\kripke_2$ be defined as the Kripke structure $(W_1\cup W_2,R_1\cup R_2,\eta_1\cup\eta_2)$. 
Without loss of generality, we will always assume that $W_1\cap W_2=\emptyset$, $R_1\cap R_2=\emptyset$, and $\eta_1\cap\eta_2=\emptyset$. 
For each $x_i$ define a Kripke structure $\kripke_{x_i}\coloneqq (W_{x_i},R_{x_i},\eta_{x_i})$, where
\begin{align*}
	W_{x_i}&\dfn \{w^{x_i}_{1},\dots,w^{x_i}_{i}\}\cup\{\,w^{x_i}_{j,1},w^{x_i}_{j,2}\mid i< j\leq n+4 \,\},\\
	R_{x_i} &\dfn\{\,(w^{x_i}_{j},w^{x_i}_{j+1})\mid 1\leq j< i\,\} \cup \{\,(w^{x_i}_i,w^{x_i}_{i+1,a})\mid a\in\{1,2\}\,\}\cup\\
	&\hspace{1.7em}\cup\{\,(w^{x_i}_{j,a},w^{x_i}_{j+1,a})\mid a\in\{1,2\}, i< j< n+4\,\}\cup\{\,(w^{x_i}_{n+4,a},w^{x_i}_{n+4,a})\mid a\in\{1,2\}\,\}\\
	\eta_{x_i}&\dfn\{\,(w,\{x_1,\dots,x_n\})\mid w\in\{w^{x_i}_{n+3,1},w^{x_i}_{n+4,2}\}\,\}\cup\{\,(w,\{x_1,\dots,x_n\}\setminus\{x_i\})\mid w\in\{w^{x_i}_{n+4,1},w^{x_i}_{n+3,2}\}\,\},
\end{align*}
and for all other not yet defined $\eta_{x_i}(w)$ with $w\in W_{x_i}$, we let $\eta_{x_i}(w)=\emptyset$.
 
If $\ell$ is a literal, then we will write $\var(\ell)$ to denote the corresponding variable of $\ell$. 
Furthermore, we define a Kripke structure $\kripke_\varphi\coloneqq (W_\varphi,R_\varphi,\eta_\varphi)$, where
\begin{align*}
	W_\varphi &\dfn \{w^{c}_{i}\mid 1\leq i\leq n+1\}\cup\{w^{c_j}\mid 1\leq j\leq m\} \cup\{w^{c_j}_{j,i,k}\mid 1\leq j\leq m,1\leq i\leq 3,1\leq k\leq 2\},\\
	R_\varphi &\dfn \{\,(w^{c}_{i},w^{c}_{i+1})\mid 1\leq i \leq n\,\}\cup\{\,(w^c_{n+1},w^{c_j})\mid 1\leq j\leq m\,\}\cup\\
	&\hspace{1.7em} \cup \{(w^{c_j},w^{c_j}_{j,i,1}),(w^{c_j}_{j,i,1}\,,w^{c_j}_{j,i,2})\mid 1\leq\! i\!\leq 3,1\leq\!j\!\leq m\}\cup\\
	&\hspace{1.7em} \cup \{(w^{c_j}_{j,i,2}\,,w^{c_j}_{j,i,2})\mid 1\leq\! i\!\leq 3,1\leq\!j\!\leq m\},\text{ and}\\
	\eta_\varphi &\dfn \big\{(w^{c_j}_{j,i,1}\,, \{x_k\mid \ell_{j,i}=x_k\}\cup\{x_k\mid x_k\neq\var(\ell_{j,i})\})\;\big|\; 1\leq j\leq m,1\leq i\leq 3\big\}\cup\\
	&\hspace{1.7em} \cup \big\{(w^{c_j}_{j,i,2}\,, \{x_k\mid \ell_{j,i}=\lnot x_k\}\cup\{x_k\mid x_k\neq\var(\ell_{j,i})\})\;\big|\; 1\leq j\leq m,1\leq i\leq 3\big\},
\end{align*}
and for all other not yet defined $\eta_\varphi(w)$ for $w\in W_{\varphi}$, we let $\eta_\varphi(w)=\emptyset$. 

Finally, let $\kripke=(W,R,\eta)$ be the Kripke structure defined as 
\[\bigcup_{1\leq i\leq n}\kripke_{x_i}\cup\kripke_\varphi.\]
Furthermore, set 
\begin{align*}
  T\dfn\mset{w^{x_{1}}_{1},\dots,w^{x_{n}}_{1},w_{1}^{c}} \text{ and }
 \psi\dfn\underbrace{\EX\AX\cdots\mathsf{P}\X}_n\AX\EX\bigwedge_{i=1}^{n}\EF x_{i},
\end{align*}
where $\mathsf{P}=\E$ if $n$ is odd and $\mathsf{P}=\A$ if $n$ is even. 
We define the reduction as $f\colon\langle\varphi\rangle\mapsto\langle \kripke,T,\psi\rangle$. 

Figure~\ref{fig:example-pspace} shows an example of the reduction for the instance \[\exists x_{1}\forall x_{2}\exists x_{3}(x_{1}\lor\lnot{x_{2}}\lor\lnot{x_{3}})\land(\lnot{x_{1}}\lor x_{2}\lor x_{3})\land(\lnot{x_{1}}\lor\lnot{x_{2}}\lor\lnot{x_{3}}).\] Note that this formula is a valid $\QBF$ and hence belongs to $\QBFVAL$. 
The left three branching systems choose the values of the $x_{i}$'s. 
A decision for the left (right, resp.) path should be understood as setting the variable~$x_{i}$ to the truth value 1 (0, resp.).

\begin{figure*}
	\Description[Construction for hardness proof.]{}
	\centering
	\resizebox{\linewidth}{!}{\input{general-structure-pspace-hardness}}
	
	\caption{General view on the created Kripke structure in the proof of Lemma~\ref{lem:mcs-pspacehard}. For $\kripke_{x_i}$ choosing the left path of the structure corresponds to setting $x_i$ to 1, choosing the right path to setting $x_i$ to 0. Synchronicity of the variable Kripke structures $\kripke_{x_i}$ and the structure $\kripke_\varphi$ ensures choosing consistent values for the $x_i$'s while satisfying all clauses.}
	\label{fig:generalview-pspace-struc}
	\end{figure*}

For the correctness of the reduction we need to show that $\varphi\in\QBFVAL$ iff $f(\varphi)\in\TMC(\TeamCTL)$.

``$\Rightarrow$'': Let $\varphi\in\QBFVAL$, $\varphi=\exists x_{1}\forall x_{2}\cdots\Game x_{n}\chi$, where $\chi=\bigwedge_{j=1}^{m}\bigvee_{i=1}^{3}\ell_{j,i}$. 
Thus, there exists a subtree $S$ of the full binary assignment tree on the variable set, obtained by encoding the existential choices of the quantifier prefix $\exists x_{1}\forall x_{2}\cdots\Game x_{n}$, such that $s\models \chi$ for every leaf $s$ in the subtree $S$. 
The tree $S$ specifies the choices for each existential variable depending on the preceeding universal choices.

Now we will prove that $\kripke,T\models\psi$.
For $w_{1}^{c}$ there is no choice in the next $n$ steps defined by the prefix of $\varphi$. 
For $w_{1}^{x_{1}},\dots,w_{1}^{x_{n}}$, after $n$ steps for each $i$ one successor of $w^{x_i}_{n+1,1}$ or $w^{x_i}_{n+1,2}$ is determined depending on the subtree $S$ and the to $x_i$ related path to a leaf of the assignment subtree in $S$, i.e., the corresponding assignment $s$ tells which value is assigned to $x_i$: 
if $x_i$ is mapped to $1$ the choice is the successor world $w^{x_{i}}_{i+1,1}$;
if $x_{i}$ is mapped to $0$ the choice is $w^{x_{i}}_{i+1,2}$ instead. 

Now fix an arbitrary path in the subtree $S$ which specifies a particular leaf $s$. 
Let us write $s(x)$ for the value of $x$ in the leaf $s$. 
Furthermore, recall that $s\models\chi$, so all clauses in $\chi$ are satisfied by $s$. 
Let us start with the evaluation of $\kripke,T\models\psi$.
After $n$ steps, the current team $T'$ then is $\mset{w^{c}_{n+1}}\cup\mset{w^{x_{i}}_{n+1,1}\mid s(x_{i})=1,1\leq i\leq n}\cup\mset{w^{x_{i}}_{n+1,2}\mid s(x_{i})=0,1\leq i\leq n}$ (note that now the team agrees with the assignment in $s$). 
In the next step, an $\AX$ in $\psi$ has to be evaluated. 
Hence, from any $T'$-compatible function $f$ the successor team $T'[f,1]$ has to be considered. 
Choose an arbitrary such $f$ which maps $w^{c}_{n+1}$ to some $w^{c_{j}}$ for $1\leq j\leq m$ (meaning that any clause in $\chi$ has to be satisfied) while for the $\kripke_{x_i}$-vertices one proceeds one step to $w^{x_i}_{n+2,k}$ depending on the respective $k\in\{1,2\}$. 
As a result, $T'[f,1]$ is $\mset{w^{c_{j}}}\cup\mset{w^{x_{i}}_{n+2,1}\mid s(x_{i})=1,1\leq i\leq n}\cup\mset{w^{x_{i}}_{n+2,2}\mid s(x_{i})=0,1\leq i\leq n}$. 
Now continuing with an $\EX$ in $\psi$, the team member $w^{c_j}$ has to select a compatible function $f$ which chooses some successor world. 
This choice is made according to the assignment $s$ under consideration. 
From $s\models\chi$, we know that for the $j$th clause $C_{j}$, we have that $s\models C_{j}$, hence, at least one literal is satisfied in that clause.
W.l.o.g.\ assume that in clause $C_{j}$ the literal $\ell_{j,r}$ (for some $r\in\{1,2,3\}$) is satisfied by $s$, i.e., $s\models\ell_{j,r}$. 
Then, the function $f$ maps to the world $w^{c_{j}}_{j,r,1}$ as a successor from $w^{c_{j}}$. 
For the $\kripke_{x_i}$ vertices $w^{x_{i}}_{n+2,k}$ with $k\in\{1,2\}$, we again have no choice and proceed to $w^{x_{i}}_{n+3,k}$. 

Now, the remainder of $\psi$ which is $\bigwedge_{i=1}^{n}\EF x_{i}$ has to be satisfied. 
For the purpose of the proceeding proof, fix an $1\leq i\leq n$ and consider each case $\EF x_i$ separately. 
Observe that for the variable team members $w^{x_{i}}_{n+3,1}$ (i.e., those corresponding to variables that are set to $1$ by $s$) all variables are in the label of the current world, but in the successor world $w^{x_{i}}_{n+4,1}$ the variable $x_i$ is not in the label, i.e., $x_{i}\notin \eta(w^{x_{i}}_{n+4,1})$. 
The opposite is true for the $w^{x_{i}}_{n+3,2}$ worlds (i.e., those corresponding to variables that are set to $0$ by $s$) where $x_{i}\notin \eta(w^{x_{i}}_{n+3,2})$ but all variables are in the label of the successor world, i.e., $x_{i}\in \eta(w^{x_{i}}_{n+4,2})$. 
Depending on what kind of literal $\ell_{j,r}$ is, the corresponding variable $x_{i'}$ is either in the label of the world $w^{c_j}_{j,r,1}$ or in the one of the world $w^{c_j}_{j,r,2}$. 
First, assume that $i=i'$. 
If $\ell_{j,r}=x_i$, then $x_i\in\eta(w^{c_j}_{j,r,1})$ and if $\ell_{j,r}=\lnot x_i$, then $x_i\in\eta(w^{c_j}_{j,r,2})$. 
As $s\models\ell_{j,r}$, it follows that if $\ell_{j,r}=x_i$, then $s(x_i)=1$ and hence $x_i\in\eta(w^{c_j}_{j,r,1})$ and $x_i\in\eta(w^{x_i}_{n+3,1})$ (similarly, if $\ell_{j,r}=\lnot x_i$, then $s(x_i)=0$ and hence $x_i\in\eta(w^{c_j}_{j,r,2})$ and $x_i\in\eta(w^{x_i}_{n+4,2})$). 
Thus, in the first case, $\EF x_i$ is satisfied by choosing the $k$ in the $\EF$ to be $0$ and thereby proceeding in $w^{c_j}_{j,r,1}$, resp., in $w^{x_i}_{n+3,1}$, i.e., $T[f,0]$ is chosen to satisfy $x_i$. 
This is correct, as in this case $x_i$ is in the label of these worlds. 
For the second case, the choice is $T[f,1]$, hence having $k=1$ in the $\EF$-evaluation, and this is also correct, as in this case $x_i$ is in the label of the respective successor worlds $w^{x_i}_{n+4,2}$ and $w^{c_j}_{j,r,2}$. 
In the other case, $i\neq i'$, it is irrelevant for the satisfaction of $\EF x_i$ as $x_{i'}$ is in the labels of both worlds.

As we considered an arbitrary path in $S$, and arbitrary functions $f$ for the branching choices from $w^c_{n+1}$, it follows that $\kripke,T\models\psi$.

``$\Leftarrow$'': Let $\kripke,T\models\psi$. 
Then, by the semantics of $\psi$, there is a sequence of $n$ choices for the $\EX$ and $\AX$ operators in the prefix of $\psi$ such that for any choice of the $\AX$ operators, the formula $\AX\EX\bigwedge_{i=1}^{n}\EF x_{i}$ is satisfied in the end. 
Let us denote by a binary vector $s\in\{0,1\}^{n}$ a possible sequence of choices for the $\EX$/$\AX$ operators in the prefix of $\psi$. 
To be more precise, if the $i$th choice is to proceed to $w^{x_i}_{i+1,1}$, then set the $i$th bit of $s$ to $s_i=1$ and otherwise (where one proceeds to $w^{x_i}_{i+1,2}$) to $s_i=0$. 
Irrespectively how the choices for the $\AX$ operators are made, in the end $\AX\EX\bigwedge_{i=1}^{n}\EF x_{i}$ has to be satisfied by assumption. 

The next $\AX$-choice considers all $w^{c_j}$ for $1\leq j\leq m$. 
Then, the next $\EX$-choice considers a $w^{c_j}_{j,r,1}$ for $1\leq r\leq 3$. 
Finally, the formula $\bigwedge_{i=1}^{n}\EF x_{i}$ has to be satisfied. 
These $n$ $\EF$-choices have to be made consistently with respect to the states $w^{x_i}_{n+3,1}$ and $w^{x_i}_{n+3,2}$ for $1\leq i\leq n$ (which have been decided before already by the choices encoded in the vector $s$). 
Here, consistent means that choices that lead to $w^{x_i}_{n+3,1}$ have to choose $k=0$ in the $\EF$ operator and those that lead to $w^{x_i}_{n+3,2}$ have to choose $k=1$ in the $\EF$ operator. 
This is due to the labeling of these worlds. 
From $s$ we can construct an assignment $s'$ in the obvious way: $s'(x_i)=s_i$ for all $1\leq i\leq n$. 
This assignment $s'$ is indeed well-defined. 
Moreover, depending on the $w^{c_j}_{j,r,1}$ (for $1\leq r\leq 3$) that has been decided, we can deduce that the corresponding literal $\ell_{j,r}$ is satisfied by $s'$.
As $\bigwedge_{i=1}^{n}\EF x_{i}$ is satisfied, irrespective on the choice of $w^{c_j}$, we can deduce that $s'\models\bigvee_{j=1}^{3}\ell_{i,j}$ for all $1\leq j\leq m$. 
From this we can conclude that $s'\models\chi$, and as $s$ was determined according to the $\EX/\AX$-prefix of length $n$ of $\psi$, finally, $\varphi\in\QBFVAL$. 
\end{proof}

\begin{figure}
	\Description[An example structure for the construction in the proof.]{}
	\centering
\begin{tikzpicture}[c/.style={circle,fill=black,inner sep=0mm,minimum width=1.5mm},x=.7cm,y=.8cm]

 \foreach \x/\y/\n/\l in {0/0/0/,-.5/-1/l1/,.5/-1/r1/,-.5/-2/l2/,.5/-2/r2/,-.5/-3/l3/,.5/-3/r3/,-.5/-4/l4/,.5/-4/r4/,%
 	-.5/-5/l5/$\substack{x_{1}\\x_{2}\\x_{3}}$,.5/-5/r5/$\substack{x_{2}\\x_{3}}$,%
	-.5/-6/l6/$\substack{x_{2}\\x_{3}}$,.5/-6/r6/$\substack{x_{1}\\x_{2}\\x_{3}}$}
  \node[c,label={[xshift=1.3mm]180:\footnotesize\l}] (a\n) at (\x,\y) {};

 \foreach \x/\y/\n/\l in {2.5/0/0/,2.5/-1/l1/,2/-2/l2/,3/-2/r2/,2/-3/l3/,3/-3/r3/,2/-4/l4/,3/-4/r4/,%
 	2/-5/l5/$\substack{x_{1}\\x_{2}\\x_{3}}$,3/-5/r5/$\substack{x_{1}\\x_{3}}$,%
	2/-6/l6/$\substack{x_{1}\\x_{3}}$,3/-6/r6/$\substack{x_{1}\\x_{2}\\x_{3}}$}
  \node[c,label={[xshift=1.3mm]180:\footnotesize\l}] (b\n) at (\x,\y) {};

 \foreach \x/\y/\n/\l in {5/0/0/,5/-1/l1/,5/-2/l2/,4.5/-3/l3/,5.5/-3/r3/,4.5/-4/l4/,5.5/-4/r4/,%
 	4.5/-5/l5/$\substack{x_{1}\\x_{2}\\x_{3}}$,5.5/-5/r5/$\substack{x_{1}\\x_{2}}$,%
	4.5/-6/l6/$\substack{x_{1}\\x_{2}}$,5.5/-6/r6/$\substack{x_{1}\\x_{2}\\x_{3}}$}
  \node[c,label={[xshift=1.3mm]180:\footnotesize\l}] (c\n) at (\x,\y) {};

 \foreach \f/\t in {0/l1,0/r1,l1/l2,l2/l3,l3/l4,l4/l5,l5/l6,r1/r2,r2/r3,r3/r4,r4/r5,r5/r6}{
  \path[-stealth',black] (a\f) edge (a\t);
 }
 \foreach \f/\t in {0/l1,l1/l2,l1/r2,l2/l3,l3/l4,l4/l5,l5/l6,r2/r3,r3/r4,r4/r5,r5/r6}{
  \path[-stealth',black] (b\f) edge (b\t);
 }
 \foreach \f/\t in {0/l1,l1/l2,l2/l3,l2/r3,l3/l4,l4/l5,l5/l6,r3/r4,r4/r5,r5/r6}{
  \path[-stealth',black] (c\f) edge (c\t);
 }

 \foreach \x/\y/\n/\l in 
 {12/0/0/,
 12/-1/1/,
 12/-2/2/,
 12/-3/3/,
 8.5/-4/l4/,
 12/-4/m4/,
 15.5/-4/r4/,%
 7.5/-5/ll5/$\substack{x_{1}\\x_{2}\\x_{3}}$,
 8.5/-5/lm5/$\substack{x_{1}\\x_{3}}$,
 9.5/-5/lr5/$\substack{x_{1}\\x_{2}}$,
 7.5/-6/ll6/$\substack{x_{2}\\x_{3}}$,
 8.5/-6/lm6/$\substack{x_{1}\\x_{2}\\x_{3}}$,
 9.5/-6/lr6/$\substack{x_{1}\\x_{2}\\x_{3}}$,
 11/-5/ml5/$\substack{x_{2}\\x_{3}}$,
 12/-5/mm5/$\substack{x_{1}\\x_{2}\\x_{3}}$,
 13/-5/mr5/$\substack{x_{1}\\x_{2}\\x_{3}}$,
 11/-6/ml6/$\substack{x_{1}\\x_{2}\\x_{3}}$,
 12/-6/mm6/$\substack{x_{1}\\x_{3}}$,
 13/-6/mr6/$\substack{x_{1}\\x_{2}}$,
 14.5/-5/rl5/$\substack{x_{2}\\x_{3}}$,
 15.5/-5/rm5/$\substack{x_{1}\\x_{3}}$,
 16.5/-5/rr5/$\substack{x_{1}\\x_{2}}$,
 14.5/-6/rl6/$\substack{x_{1}\\x_{2}\\x_{3}}$,
 15.5/-6/rm6/$\substack{x_{1}\\x_{2}\\x_{3}}$,
 16.5/-6/rr6/$\substack{x_{1}\\x_{2}\\x_{3}}$}
  \node[c,label={[xshift=1.3mm]180:\footnotesize\l}] (d\n) at (\x,\y) {};

 \foreach \f/\t in {0/1,1/2,2/3,3/l4,3/m4,3/r4,%
 	l4/lm5,ll5/ll6,lm5/lm6,lr5/lr6,%
	m4/mm5,ml5/ml6,mm5/mm6,mr5/mr6,%
	r4/rm5,rl5/rl6,rm5/rm6,rr5/rr6}{
  \path[-stealth',black] (d\f) edge (d\t);
 }
 
 \foreach \x in {al6,ar6,bl6,br6,cl6,cr6,dll6,dlm6,dlr6,dml6,dmm6,dmr6,drl6,drm6,drr6}{ 
  \path[-stealth',black] (\x) edge[>=stealth',loop below,] (\x);
 }
 
 \foreach \f/\t in {l4/ll5,m4/ml5,r4/rl5}{
  \path[-stealth',black] (d\f) edge[bend right] (d\t);
 }
 
 \foreach \f/\t in {l4/lr5,m4/mr5,r4/rr5}{
  \path[-stealth',black] (d\f) edge[bend left] (d\t);
 }
 
 \draw[black,dashed,rounded corners] (-.3,-.3) rectangle (12.3,.3);
 \node[anchor=west] at (12.5,0) {Team $T$};
 
 \node[text width=2cm,align=center] at (8,-3) {\footnotesize agreed assignment};
 \draw[black,dashed] (-1,-3) -- (7.5,-3);
 
\end{tikzpicture}
 \caption{Example structure built in proof of Lemma~\ref{lem:mcs-pspacehard} for the QBF $\exists x_1\forall x_2\exists x_3(x_1\lor\lnot x_2\lor \lnot x_3)\land(\lnot x_1\lor x_2\lor x_3)\land(\lnot x_1\lor \lnot x_2\lor \lnot x_3)$.}\label{fig:example-pspace}
\end{figure}

Now, we will turn towards proving the $\PSPACE$ upper bound for $\TMC(\TeamCTL)$. 
Before, we will need a definition and two auxiliary lemmas. 
Given two teams $T_1,T_2$ of a Kripke structure $\kripke$, we say $T_2$ is a \emph{successor team} of $T_1$ if there exists a $T_1$-compatible function $f$ such that $T_2\meq T_1[f,1]$.

\begin{lemma}\label{lem:aux-number-of-successor-teams}
	Let $\kripke$ be a Kripke structure and $T$ be a team of $\kripke$. 
	Furthermore, let $T_1,T_2,\dots,T_m$ be a sequence of teams such that $T_1 = T$ and  such that $T_{i+1}$ is a successor team of $T_i$ for all $i$. 
	Then the largest $m$ such that there exists no $1\leq i\neq j\leq m$ with $T_i\meq T_j$ is bounded from above by $|W|^{|T|}$. 
	In particular, if there is such a sequence of length $|W|^{|T|}+1$, then there is an infinite sequence of this kind.
\end{lemma}
\begin{proof}
	The number of teams of $\mathcal K$ with cardinality $|T|$ is bounded by the number of functions of the kind $f\colon T\to W$. 
	As $|W|^{|T|}$ is the number of all functions from $T$ to $W$, the claim follows.
\end{proof}

The following lemma shows that the successor relation over teams can be decided in polymonial time. Here, we measure the size of a team by the cardinality of the multiset and the size of the Kripke structure by the number of its worlds.

\begin{lemma}\label{lem:succ-check-for-teams-in-p}	
The question whether a given team~$T_2$ is a successor team of a given team~$T_1$ (w.r.t.\ a given Kripke structure~$\kripke$) can be decided in polynomial time in $|T_1|+|T_2|+|\kripke|$.
\end{lemma}
\begin{proof}
	Let $T_1=\mset{t_{1},\dots,t_{n}}$ and $T_2=\mset{t_{1}',\dots,t_{n}'}$. 
	Then, $T_2$ is a successor team of $T_1$ if and only if the following is true
	\begin{enumerate}
		\item $|T_1|=|T_2|$, 
		\item for all $w\in T_1$ there exists a $w'\in T_2$ such that $wRw'$, and
		\item for all $w\in T_2$ there exists a $w'\in T_1$ such that $w'Rw$.
	\end{enumerate}
	The first two items alone do not suffice because some $z$ in $T_2$ could be reached from two different $y$'s in $T_1$.
	The last item ensures that there is no $x$ in $T_2$ that has no $R$-predecessor in $T_1$.

	The first item can be checked in time $O(|T_1|)$. 
	The second and third item can both be checked in polynomial time in~$|T_1|\cdot |T_2|\cdot|\kripke|$. 
	Together this is polynomial time in the input length.
\end{proof}

\SetKw{proce}{Procedure}
	\SetKwFunction{pPSearch}{EG-PathSearch}
	\SetKwFunction{pmc}{MC}
	\LinesNumbered
	\begin{algorithm}\caption{Algorithm for EG-PathSearch Procedure.}\label{algo:PathSearch}
		\small
		\proce \pPSearch{Kripke structure $\kripke$, team $T_1$, team $T_2$ with $|T_1|=|T_2|$, formula $\varphi$, integer $c$}\;
		\lIf{$c=0$}{\Return $(T_1\meq T_2)\land \pmc(\kripke,T _1,\varphi)$}
		\lIf{$c=1$}{\Return $(T_2$ is a successor team of $T_1)\land \pmc(\kripke,T _1,\varphi)\land \pmc(\kripke,T _2,\varphi)$}
		existentially branch on all teams~$T_{\text{mid}}$ of $\kripke$ with $|T_{\text{mid}}|=|T_1|=|T_2|$ and\linebreak \Return $\pPSearch(\kripke,T _1,T_{\text{mid}},\varphi,\lfloor c/2\rfloor)\land\pPSearch(\kripke,T _{\text{mid}},T_2,\varphi,\lceil c/2\rceil)$
	\end{algorithm}

\SetKw{proce}{Procedure}
	\SetKwFunction{pnPSearch}{AF-PathSearch}
	\LinesNumbered
	\begin{algorithm}\caption{Algorithm for AF-PathSearch Procedure.}\label{algo:NegPathSearch}
		\small
		\proce \pnPSearch{Kripke structure $\kripke$, team $T_1$, team $T_2$ with $|T_1|=|T_2|$, formula $\varphi$, integer $c$}\;
		\lIf{$c=0$}{\Return $(T_1\not\meq T_2)\varovee \pmc(\kripke,T _1,\varphi)$}
		\lIf{$c=1$}{\Return $(T_2$ is \textbf{not} a successor team of $T_1)\varovee \pmc(\kripke,T _1,\varphi)\varovee \pmc(\kripke,T _2,\varphi)$}
		universally branch on all teams~$T_{\text{mid}}$ of $\kripke$ with $|T_{\text{mid}}|=|T_1|=|T_2|$ and\linebreak \Return $\pnPSearch(\kripke,T _1,T_{\text{mid}},\varphi,\lfloor c/2\rfloor)\varovee\pnPSearch(\kripke,T_{\text{mid}}, T_2,\varphi,\lceil c/2\rceil)$
	\end{algorithm}

	\LinesNumbered
	\begin{algorithm}\caption{Algorithm for $\TMC(\TeamCTL)$.}\label{algo:TMC-CTL}
	\small
	\proce \pmc{Kripke structure $\kripke$, team $T$, formula $\varphi$}\;
			\If{$\varphi = p$ (resp., $\varphi=\lnot p$) and $p\in\PROP$}{
				\leIf{for all $w\in T$ we have that $p\in\eta(w)$ (resp., $p\notin\eta(w)$)}{\Return true}{\Return false}
			}
			\lIf{$\varphi=\psi_1\land\psi_2$}{
				\Return $\pmc(\kripke,T ,\psi_1)\land\pmc(\kripke,T ,\psi_2)$
			}
			\If{$\varphi=\psi_1\lor\psi_2$}{
				existentially branch on all $T_{1}\cup T_{2}=T$ s.t.\ $T_1\cap T_2=\emptyset$\; 
				\Return $\pmc(\kripke,T _{1},\psi_1)\land\pmc(\kripke,T _{2},\psi_2)$
			}
			\If{$\varphi=\EX\psi$}{
				existentially branch on all successors $T'$ of $T$\tcp*[r]{successor check is in $\P$ by Lem.~\ref{lem:succ-check-for-teams-in-p}}
				\Return $\pmc(\kripke,T ',\psi)$
			}
			\If{$\varphi=\AX\psi$}{
				universally branch on all successors $T'$ of $T$\tcp*[r]{successor check is in $\P$ by Lem.~\ref{lem:succ-check-for-teams-in-p}}
				\Return $\pmc(\kripke,T ',\psi)$
			}
			\If{$\varphi=\E[\alpha\U\beta]$}{
				existentially branch on path length $c\in[0,|W|^{|T|}]$ \tcp*[r]{Lem.~\ref{lem:aux-number-of-successor-teams}}
				\lIf{$c=0$}{\Return $\pmc(\kripke,T ,\beta)$}
				\Else{
					existentially branch on $T_{\text{end}-1},T_{\text{end}}$ s.t.\ $T_{\text{end}}$ is a successor of $T_{\text{end}-1}$ and $|T|=|T_{\text{end}}|=|T_{\text{end}-1}|$\;
					\Return $\pPSearch(\kripke,T ,T_{\text{end}-1},\alpha,c-1)\land\pmc(\kripke,T _{\text{end}},\beta)$
				}
			}
			\If{$\varphi=\E[\alpha\R\beta]$}{
				existentially branch on $\{\texttt{no-}\alpha, \texttt{some-}\alpha\}$\;
				\If{\texttt{no-}$\alpha$}{
					existentially branch on $T_{\text{end}}$ with $|T|=|T_{\text{end}}|$\; 
					\Return $\pPSearch(\kripke,T ,T_{\text{end}},\beta,|W|^{|T|}+1)$\tcp*[f]{Lem.~\ref{lem:aux-number-of-successor-teams}, infinite sequence}
				}
				\ElseIf{\texttt{some-}$\alpha$}{
					existentially branch on path length $c\in[0,|W|^{|T|}]$ \tcp*[r]{Lem.~\ref{lem:aux-number-of-successor-teams}}
						existentially branch on $T_{\text{end}}$ with $|T|=|T_{\text{end}}|$\; 
						\Return $\pmc(\kripke,T_{\text{end}},\alpha)\land\pPSearch(\kripke,T, T_{\text{end}},\beta,c)$
				}
			}
			\If{$\varphi=\A[\alpha\U\beta]$}{
				universally branch on $\{\texttt{reach-}\beta,\texttt{all-}\alpha\texttt{-before-first-}\beta\}$\;
				\If{$\texttt{reach-}\beta$}{
					universally branch on $T_{\text{end}}$ with $|T|=|T_{\text{end}}|$\; 
					\Return $\pnPSearch(\kripke,T ,T_{\text{end}},\beta,|W|^{|T|}+1)$\tcp*[f]{Lem.~\ref{lem:aux-number-of-successor-teams} implies bounded sequence}
				}
				\ElseIf{$\texttt{all-}\alpha\texttt{-before-first-}\beta$}{
					universally branch on path length $c\in[0,|W|^{|T|}]$ \tcp*[r]{Lem.~\ref{lem:aux-number-of-successor-teams}}
						universally branch on $T_{\text{end}}$ with $|T|=|T_{\text{end}}|$\;
						\Return $\pmc(\kripke,T_{\text{end}},\alpha)\varovee\pmc(\kripke,T_{\text{end}},\beta)\varovee\pnPSearch(\kripke,T, T_{\text{end}},\beta,c)$
				}

			}
			\If{$\varphi=\A[\alpha\R\beta]$}{
				universally branch on path length $c\in[0,|W|^{|T|}]$ \tcp*[r]{Lem.~\ref{lem:aux-number-of-successor-teams} implies longer sequences}
				\lIf{$c=0$}{\Return $\pmc(\kripke,T ,\beta)$}
				\Else{universally branch on $T_{\text{end}-1},T_{\text{end}}$ s.t.\ $T_{\text{end}}$ is a successor of $T_{\text{end}-1}$ and $|T|=|T_{\text{end}}|=|T_{\text{end}-1}|$\;  
				\Return $\pnPSearch(\kripke,T ,T_{\text{end}-1},\alpha,c-1)\varovee\pmc(\kripke,T _{\text{end}},\beta)$}
			}
	\end{algorithm}	

\begin{lemma}\label{lem:smcpspace}
 $\TMC(\TeamCTL)$ is in $\PSPACE$.
\end{lemma}
\begin{proof}
	We construct an algorithm that runs in alternating polynomial time. As $\mathrm{APTIME}=\mathrm{PSPACE}$~\cite{ChandraKS81}, this proves the desired upper bound.
	
	Before we come to the model checking algorithm, we want to shortly discuss two subroutines that are used in it.
	The first subroutine (depicted in Algorithm~\ref{algo:PathSearch}) is used in the cases for the existential binary temporal operators, and also makes use of the main algorithm $\pmc$ (depicted in Algorithm~\ref{algo:TMC-CTL}):
	The procedure \pPSearch will recursively determine whether there exists a path of length~$c$ between two given teams such that a given formula~$\varphi$ is satisfied at each team of the path.
	Intuitively, the procedure works similarly as the path search method in the proof of Savitch's Theorem~\cite{DBLP:conf/stoc/Savitch69}
(also, see Sipser's textbook~\cite[Section~8.1]{sipser13}).
	It will be called for path lengths $c$ that are bounded from above by $|W|^{|T|}+1$ (Lemma~\ref{lem:aux-number-of-successor-teams}).
	As every recursion halves the path length, the recursion depth is bounded by $O(\log(|W|^{|T|}))=O(|T|\log|W|)$ and hence is polynomial in the input length.
	
	The second subroutine (depicted in Algorithm~\ref{algo:NegPathSearch}) is used for the universal binary temporal connectives, and also makes  use of the main algorithm~$\pmc$. Intuitively, it is the dual of Algorithm~\ref{algo:PathSearch}, i.e., it will recursively determine whether all paths of length~$c$ between two given teams are such that a given formula~$\varphi$ is satisfied at at least one team of the path. 
	Thus, the subroutine will reject if and only if there is a path of length~$c$ between the two teams such that at all teams of this path $\varphi$ is \emph{falsified}.	
	Again, it will be called for path lengths $c$ that are bounded from above by $|W|^{|T|}+1$, which implies that the recursion depth is polynomial in the input length.

	Next, we define the main algorithm.
	Let $\varphi$ be a \TeamCTL-formula, $\kripke$ be a Kripke structure, and $T$ be a team of $\kripke$. 
	Given these, the algorithmic call $\texttt{MC}(\kripke,T ,\varphi)$ returns true if and only if $\kripke,T \models\varphi$. 
	The algorithm is depicted in Algorithm~\ref{algo:TMC-CTL}.
	
	Intuitively, the algorithm is a recursive alternating tableaux algorithm that branches on the operators of the subformulae of $\varphi$. 
	For the existential temporal operators, we are searching (using nondeterminism) for paths witnessing the satisfaction while for the universal temporal operators, we are searching (using universal branching) for the absence of paths witnessing violation of the formula.
	Note that we do not allow classical negations in our logic, hence we use duality of alternating Turing machines to handle negations.
	
	Observe that the recursion depth of the main algorithm is bounded by the number of subformulae; therefore, it is linear in the input length.
	The recursion depth of the calls to the subroutines is also bounded polynomially, as argued above.
	Also note that the size of each guessed value is polynomial in the input size; for the value of $c$ notice that it is encoded in binary. 
	
	As for the correctness of the algorithm, each step directly implements the corresponding semantics. 
	Several times we rely on the fact that the successor team check is in $\P$ for the guessed teams (see Lemma~\ref{lem:succ-check-for-teams-in-p}), i.e., in lines 9, 12, 18, and 42.
	All cases but the four for $\U$ and $\R$ are straightforward applications of the semantics, so let us focus on the latter.
	
	We begin our explanations with the existential variants.
	Here, both the $\U$ and $\R$ case make use of Lemma~\ref{lem:aux-number-of-successor-teams} and thereby restrict the path length to an exponential value: if there is a path witnessing satisfaction of a subformula, then there is also one of at most exponential length. To find such paths, we rely on Algorithm~\ref{algo:PathSearch}.
	
	For the $\U$ case, the algorithm checks whether $\beta$ holds directly at $T$ (line 16) or if there 
	are teams~$T_{\text{end}}$ and $T_{\text{end}-1}$ such that $T_{\text{end}}$ is a successor of $T_{\text{end}-1}$, $\beta$ holds in $T_{\text{end}}$, and there is a path from $T$ to $T_{\text{end}-1}$ where $\alpha$ holds at every position (lines 18 and 19). Such a path can be assumed to have at most exponential length due to Lemma~\ref{lem:aux-number-of-successor-teams}.
	This captures exactly the semantics of the existential until. 
	
	For the $\R$ case, the algorithm checks one of two cases. Either (lines 23 and 24), there is an exponentially long path on which $\beta$ holds everywhere (i.e., $\alpha$ must not hold), which implies due to Lemma~\ref{lem:aux-number-of-successor-teams} that there is also in infinite such path. Or (lines 26 to 28), there is a team~$T_{\text{end}}$ such that $\alpha$ holds at $T_{\text{end}}$ and there is a path from $T$ to $T_{\text{end}}$ where $\beta$ holds everywhere, including at $T_{\text{end}}$. Such a path can again be assumed to have at most exponential length due to Lemma~\ref{lem:aux-number-of-successor-teams}.
	The disjunction of these two cases captures exactly the semantics of the existential release.
	
	Let us finish by considering the universal variants of the temporal operators.
	Here, we again rely on Lemma~\ref{lem:aux-number-of-successor-teams} (this time to bound the length of counterexamples, as we consider universal path quantifiers) and Algorithm~\ref{algo:NegPathSearch}.
	
	Let us begin with the $\R$ case, as it is the simpler one.
	Here, the algorithm checks whether $T$ satisfies $\beta$ (line 40) and whether for every pair~$(T_{\text{end}}, T_{\text{end}-1})$ of teams such that $T_{\text{end}}$ is a successor of $T_{\text{end}-1}$, $T_{\text{end}}$ satisfies $\beta$ or whether every path from $T$ to $T_{\text{end}-1}$ satisfies $\alpha$ somewhere (lines 42 and 43), i.e., if $\beta$ does not hold at $T_{\text{end}}$, then it must have been released by $\alpha$ on every path leading to $T_{\text{end}}$.
	Again, due to Lemma~\ref{lem:aux-number-of-successor-teams}, we can restrict ourselves here to paths of exponential length, since this property always has a counterexample of exponential length.
	This captures exactly the semantics of the universal release.
	
	Finally, for the $\U$ case, the algorithm checks two properties.
	The first one (lines 32 and 33) is that all paths of exponential length contain at least one team where $\beta$ holds. Lemma~\ref{lem:aux-number-of-successor-teams} implies that this is equivalent to stating that in all infinite paths $\beta$ eventually holds. 
	The second one (lines 35 to 37) is that for every team~$T_{\text{end}}$ at least one of the following cases holds:  $T_{\text{end}}$ satisfies  $\alpha$, or $T_{\text{end}}$ satisfies $\beta$, or all paths leading from $T$ to $T_{\text{end}}$ contain a team satisfying $\beta$.  
	This second constraint boils down to the property that in every path $\alpha$ holds forever or until $\beta$ is fulfilled.
	Again, by Lemma~\ref{lem:aux-number-of-successor-teams}, it suffices to consider only paths of exponential length for the second property.
	Thus, both properties holding captures exactly the semantics of the universal until.

	This overall guarantees polynomial runtime as well as the correctness of the alternating algorithm.
	As $\mathrm{APTIME}=\mathrm{PSPACE}$~\cite{ChandraKS81}, this proves the desired upper bound.
\end{proof}

\begin{theorem}\label{thm:TMC(CTL)-c}
 $\TMC(\TeamCTL)$ is $\PSPACE$-complete w.r.t.\ $\leqpm$-reductions.
\end{theorem}
\begin{proof}
	By Lemma~\ref{lem:smcpspace} we know that $\TMC(\TeamCTL)$ is in $\PSPACE$. 
	By Lemma~\ref{lem:mcs-pspacehard} we know that $\TMC(\TeamCTL)$ is $\PSPACE$-hard. 
	Thus, $\TMC(\TeamCTL)$ is $\PSPACE$-complete.
\end{proof}

Similar to \LTL, we can consider FO-definable atoms (see Definition~\ref{def:FO-definable}) in the \CTL-setting. Here we restrict the parameters to generalised atoms to be formulae in propositional logic. 
Analogous to the proof of Theorem~\ref{thm:GenAtoms}, we can extend the model checking algorithm of Lemma \ref{lem:smcpspace} to deal with FO-definable generalised atoms and the contradictory negation.

\begin{theorem}\label{thm:tmc-ctl-d-sim}
	Let $\mathcal{D}$ be a finite set of first-order definable generalised atoms. 
	$\TMC(\TeamCTL(\mathcal D, {\sim}))$ is $\PSPACE$-complete w.r.t.\ $\leqpm$-reductions.
\end{theorem}

\subsection{Satisfiability}

Again our benchmark here is \CTL satisfiability with classical semantics.

\begin{proposition}[\cite{fila79,pr80}]\label{prop:CTLSAT-complexity}
 Satisfiability for \CTL-formulae under classical semantics is $\EXPTIME$-complete w.r.t.\ $\leqpm$-reductions.
\end{proposition}

The same complexity result is easily transferred to \TeamCTL.

\begin{theorem}\label{thm:SAT}
 $\TSAT(\TeamCTL)$ is $\EXPTIME$-complete w.r.t.\ $\leqpm$-reductions.
\end{theorem}

\begin{proof}
 The problem merely asks whether there exists a Kripke structure $\kripke$ and a non-empty team $T$ of $\kripke$ such that $\kripke,T\models\varphi$ for given \CTL-formula~$\varphi$. 
 By downward closure (Proposition~\ref{prop:dwclos}), it suffices to check whether $\varphi$ is satisfied by a singleton team. 
 By singleton equivalence (Proposition~\ref{prop:singleton}) we then immediately obtain the same complexity bounds as for classical \CTL satisfiability. 
 Hence, the claim follows from Proposition~\ref{prop:CTLSAT-complexity}.
\end{proof}

\section{Conclusions}

We introduced and studied team semantics for the temporal logics \LTL and \CTL. 
We concluded that \TeamLTL (with and without generalized atoms) is a valuable logic which allows to express relevant hyperproperties and complements the expressiveness of \hyltl while allowing for computationally simpler decision problems.

For \TeamCTL, the complexity of the model checking problem increases from $\PTime$-complete for usual \CTL to $\PSPACE$-complete for \TeamCTL. This fact stems from the expressive notion of synchronicity between team members and is in line with the results of~\citeauthor{kvw00}~\cite{kvw00}.

We conclude with some directions of future work and open problems.

\subsection{Future Work for TeamLTL}

We showed that some important properties that cannot be expressed in $\hyltl$ (such as uniform termination) can be expressed in \TeamLTL. Moreover input determinism can be expressed in \TeamLTL($\dep$). Can we identify tractable \TeamLTL variants (i.e., syntactic fragments with particular dependency atoms) that can  express a rich family of hyperproperties?

We showed that with respect to expressive power, $\hyltl$ and \TeamLTL are incomparable. 
However, the expressive power of \hyltl and the different extensions of \TeamLTL introduced here is left open. 
For example, the $\hyltl$ formula $\exists \pi.p_\pi$ is expressible in $\TeamLTL(\sim)$. 

Another interesting question is whether we can characterise the expressive power of relevant extensions of \TeamLTL as has been done in first-order and modal contexts. 
Recent works have shown limits of expressivity of \TeamLTL variants via translations to extensions and fragments of $\hyltl$~\cite{DBLP:conf/fsttcs/VirtemaHFK021,DBLP:conf/foiks/KontinenSV24} and FO~\cite{DBLP:conf/wollic/KontinenS21}. 
It is also open whether the one-way translations in those papers can be strengthened to precise characterisations of expressivity as was done in the work of \citeauthor{DBLP:conf/foiks/KontinenSV24}~\cite{DBLP:conf/foiks/KontinenSV24} for asynchronous set-based \TeamLTL.

 We studied the complexity of the path-checking, model checking, and satisfiability problem of \TeamLTL and its extensions. However, many problems are still open: Can we identify matching upper and lower bounds for the missing cases and partial results of Figure \ref{fig:overview} on page \pageref{fig:overview}? In particular, what is the complexity of model-checking \TeamLTL when splitjunctions are allowed?
 
 Finally, the complexity of the validity and implication problems are open for almost all cases.

\subsection{Future Work for TeamCTL}

We only scratched the surface of the complexity of the satisfiability problem of \TeamCTL. There are two obvious directions for future work here. Complexity of \TeamCTL extended with atoms and connectives that preserve downward closure, and complexity of \TeamCTL extended with non-downward closed atoms and connectives. The complexity of the former is expected to stay relatively low and comparable to vanilla \TeamCTL, since there it still suffices to consider only singleton teams. The complexity of the latter is expected to be much higher and be perhaps even undecidable. Note that \TeamLTL with the contradictory negation is highly undecidable~\cite{DBLP:journals/tcs/Luck20}.

The tautology or validity problem for this new logic is quite interesting and seems to have a higher complexity than the related satisfiability problem. This is due to alternation of set quantification: the validity problem quantifies over teams universally while the splitjunction implements an existential set quantification. We leave the considerations related to the validity problem as future work. Formally the corresponding problem is defined as follows:
\problemdef{$\TVAL(\TeamCTL$) --- \TeamCTL Validity Problem.}{A \TeamCTL-formula $\varphi$.}{Does $\kripke,T \models\varphi$ hold for every Kripke structure $\kripke$ and every team $T$ of $\kripke$?}

In the context of team-based modal logics the computational complexity of the validity problem has been studied by~\citeauthor{DBLP:journals/iandc/Virtema17}~\cite{DBLP:journals/iandc/Virtema17},~\citeauthor{DBLP:journals/lmcs/Hannula19}~\cite{DBLP:journals/lmcs/Hannula19} and~\citeauthor{DBLP:conf/csl/Luck18}~\cite{DBLP:conf/csl/Luck18}.~\citeauthor{DBLP:journals/iandc/Virtema17} and~\citeauthor{DBLP:journals/lmcs/Hannula19} showed that the problem for modal dependence logic is $\NEXPTIME$-complete whereas~\citeauthor{DBLP:conf/csl/Luck18} established that the problem for modal logic extended with the contradictory negation is complete for the complexity class $\mathrm{TOWER}$(poly), the class of problems that can be solved in time that is bounded by some $n$-fold exponential, where $n$ itself is bounded polynomially in the input length.

It is well-known that there are several ways to measure the complexity of a model checking problem. In general, a model and a formula are given, and then one needs to decide whether the model satisfies the formula. \emph{System complexity} considers the computational complexity for the case of a fixed formula whereas \emph{specification complexity} fixes the underlying Kripke structure. We considered in this paper the \emph{combined complexity} where both a formula and a model belong to the given input. Yet the other two approaches might give more specific insights into the intractability of the model checking case we investigated. In particular, the study of the \emph{specification complexity}, where the team or the team size is assumed to be fixed, might as well be of independent interest.
 
 Finally this leads to the consideration of different kinds of restrictions on the problems. In particular for the quite strong $\PSPACE$-completeness result for model checking in team semantics it is of interest how this intractability can be explained. Hence the investigation of fragments by means of allowed temporal operators and/or Boolean operators will lead to a better understanding of this presumably untameably high complexity.

\begin{acks}
	We thank the anonymous reviewers for their valuable feedback. 
This work was supported by the DFG projects ME4279/1-1 and ME4279/3-1, and DIREC -- Digital Research Centre Denmark.
\end{acks}
\bibliographystyle{ACM-Reference-Format}
\bibliography{doi-teamltl}

%% file: general-structure-pspace-hardness.tex
\usetikzlibrary{decorations.pathreplacing}

\begin{tikzpicture}[c/.style={circle,fill=black,inner sep=0mm,minimum width=2mm},x=.85cm,y=1cm,scale=1.1]

  \begin{scope}[xshift = -1cm]
   \foreach \x/\y/\n/\l in 
    {0.25/0/0/$w_{1}^{x_{1}}$,
     -.5/-1/l1/$w_{2,1}^{x_{1}}$,
     1/-1/r1/$w_{2,2}^{x_{1}}$,
     -.5/-2/l2/$w_{3,1}^{x_{1}}$,
     1/-2/r2/$w_{3,2}^{x_{1}}$,
     -.5/-3/l3/$w_{n+1,1}^{x_{1}}$,
     1/-3/r3/$w_{n+1,2}^{x_{1}}$,
     -.5/-4/l4/$w_{n+2,1}^{x_{1}}$,
     1/-4/r4/$w_{n+2,2}^{x_{1}}$,
     -.5/-5/l5/$\substack{x_{1}\\{ \vphantom{\int\limits^x}\smash\vdots}\\x_{n}}$,
     1/-5/r5/$\substack{x_{2}\\{ \vphantom{\int\limits^x}\smash\vdots}\\x_{n}}$,
     -.5/-6/l6/$\substack{x_{2}\\{ \vphantom{\int\limits^x}\smash\vdots}\\x_{n}}$,
     1/-6/r6/$\substack{x_{1}\\{ \vphantom{\int\limits^x}\smash\vdots}\\x_{n}}$}
    \node[c,label={[xshift=1.3mm]180:\l}] (a\n) at (\x,\y) {};
  
   \foreach \f/\t in {0/l1,0/r1,l1/l2,l3/l4,l4/l5,l5/l6,r1/r2,r3/r4,r4/r5,r5/r6}{
    \path[-stealth',black] (a\f) edge (a\t);
   }
   
   \path[-stealth',dotted,black] (ar2) edge (ar3);
   \path[-stealth',dotted,black] (al2) edge (al3);
  
   \draw [decorate,decoration={brace,amplitude=5pt,mirror,raise=4pt},yshift=0pt]
  (-1.5,-6.5) -- (1.2,-6.5) node [black,midway,yshift=-1.6em] {$\kripke_{x_1}$};
  
   \foreach \x/\y/\n/\l in 
    {3.25/0/0/$w_{1}^{x_{2}}$,
     3.25/-1/l1/$w_{2}^{x_{2}}$,
     2.5/-2/l2/$w_{3,1}^{x_{2}}$,
     4/-2/r2/$w_{3,2}^{x_{2}}$,
     2.5/-3/l3/$w_{n+1,1}^{x_{2}}$,
     4/-3/r3/$w_{n+1,2}^{x_{2}}$,
     2.5/-4/l4/$w_{n+2,1}^{x_{2}}$,
     4/-4/r4/$w_{n+2,2}^{x_{2}}$,
     2.5/-5/l5/$\substack{x_{1}\\{ \vphantom{\int\limits^x}\smash\vdots}\\x_{n}}$,
     4/-5/r5/$\substack{x_{1}\\x_{3}\\{ \vphantom{\int\limits^x}\smash\vdots}\\x_{n}}$,
     2.5/-6/l6/$\substack{x_{1}\\x_{3}\\{ \vphantom{\int\limits^x}\smash\vdots}\\x_{n}}$,
     4/-6/r6/$\substack{x_{1}\\{ \vphantom{\int\limits^x}\smash\vdots}\\x_{n}}$}
    \node[c,label={[xshift=1.3mm]180:\l}] (b\n) at (\x,\y) {};

   \foreach \f/\t in {0/l1,l1/l2,l1/r2,l3/l4,l4/l5,l5/l6,r3/r4,r4/r5,r5/r6}{
    \path[-stealth',black] (b\f) edge (b\t);
   }

   \path[-stealth',dotted,black] (br2) edge (br3);
   \path[-stealth',dotted,black] (bl2) edge (bl3);
  
   \draw [decorate,decoration={brace,amplitude=5pt,mirror,raise=4pt},yshift=0pt]
  (1.5,-6.5) -- (4.2,-6.5) node [black,midway,yshift=-1.6em] {$\kripke_{x_2}$};
   
   \foreach \x in {al6,ar6,bl6,br6}{ 
    \path[-stealth',black] (\x) edge[>=stealth',loop below,] (\x);
   } 
  \end{scope}
  
   \node at (4,-1) {$\cdots$};
  
   \foreach \x/\y/\n/\l in 
    {5.75/0/0/$w_{1}^{x_{n}}$,
    5.75/-1/l1/$w_{2}^{x_{n}}$,
    5.75/-2/l2/$w_{3}^{x_{n}}$,
    5/-3/l3/$w_{n+1,1}^{x_{n}}$,
    6.5/-3/r3/$w_{n+1,2}^{x_{n}}$,
    5/-4/l4/$w_{n+2,1}^{x_{n}}$,
    6.5/-4/r4/$w_{n+2,2}^{x_{n}}$,
    5/-5/l5/$\substack{x_{1}\\{ \vphantom{\int\limits^x}\smash\vdots}\\x_{n}}$,
    6.5/-5/r5/$\substack{x_{1}\\{ \vphantom{\int\limits^x}\smash\vdots}\\x_{n\!-\!1}}$,
    5/-6/l6/$\substack{x_{1}\\{ \vphantom{\int\limits^x}\smash\vdots}\\x_{n\!-\!1}}$,
    6.5/-6/r6/$\substack{x_{1}\\{ \vphantom{\int\limits^x}\smash\vdots}\\x_{n}}$}
    \node[c,label={[xshift=1.3mm]180:\l}] (c\n) at (\x,\y) {};
  
   \foreach \f/\t in {0/l1,l1/l2,l3/l4,l4/l5,l5/l6,r3/r4,r4/r5,r5/r6}{
    \path[-stealth',black] (c\f) edge (c\t);
   }
  
   \path[-stealth',dotted,black] (cl2) edge (cr3);
   \path[-stealth',dotted,black] (cl2) edge (cl3);
  
   \draw [decorate,decoration={brace,amplitude=5pt,mirror,raise=4pt},yshift=0pt]
  (4,-6.5) -- (6.7,-6.5) node [black,midway,yshift=-1.6em] {$\kripke_{x_n}$};
  
  \draw [decorate,decoration={brace,amplitude=10pt,mirror,raise=4pt},yshift=0pt]
  (6.7,-3.2) -- (6.7,0.2) node [black,midway,xshift=0.8cm,text width=2cm,anchor=west] {quantification\newline of variables};
  
   \foreach \x/\y/\n/\l in 
   {12/0/0/$w_{1}^{c}\;\;$,
   12/-1/1/$w_{2}^{c}\;\;$,
   12/-2/2/$w_{3}^{c}\;\;$,
   12/-3/3/$w_{n+1}^{c}\;\;$,
   9.5/-4/l4/$w^{c_{1}}\;$,
   14.5/-4/r4/$w^{c_{m}}\;\;$,
   8/-5/ll5/$w_{1,1,1}^{c_{1}}$,
   9.5/-5/lm5/$w_{1,2,1}^{c_{1}}$,
   11/-5/lr5/$w_{1,3,1}^{c_{1}}$,
   8/-6/ll6/$w_{1,1,2}^{c_{1}}$,
   9.5/-6/lm6/$w_{1,2,2}^{c_{1}}$,
   11/-6/lr6/$w_{1,3,2}^{c_{1}}$,
   13/-5/rl5/$w_{m,1,1}^{c_{m}}$,
   14.5/-5/rm5/$w_{m,2,1}^{c_{m}}$,
   16/-5/rr5/$w_{m,3,1}^{c_{m}}$,
   13/-6/rl6/$w_{m,1,2}^{c_{m}}$,
   14.5/-6/rm6/$w_{m,2,2}^{c_{m}}$,
   16/-6/rr6/$w_{m,3,2}^{c_{m}}$}
    \node[c,label={[xshift=1.3mm]180:\l}] (d\n) at (\x,\y) {};
  
   \foreach \f/\t in {0/1,1/2,3/l4,3/r4,
     l4/lm5,ll5/ll6,lm5/lm6,lr5/lr6,%
    r4/rm5,rl5/rl6,rm5/rm6,rr5/rr6}{
    \path[-stealth',black] (d\f) edge (d\t);
   }
   
   \path[-stealth',dotted,black] (d2) edge (d3);

   \foreach \x in {cl6,cr6,dll6,dlm6,dlr6,drl6,drm6,drr6}{ 
    \path[-stealth',black] (\x) edge[>=stealth',loop below,] (\x);
   }
   
   \foreach \f/\t in {l4/ll5,r4/rl5}{
    \path[-stealth',black] (d\f) edge[bend right] (d\t);
   }
   
   \foreach \f/\t in {l4/lr5,r4/rr5}{
    \path[-stealth',black] (d\f) edge[bend left] (d\t);
   }
  
   \node at (12,-4) {$\cdots$};
   
   \draw[black,dashed,rounded corners] (-2.5,-.3) rectangle (12.3,.3);
   \node[anchor=west] at (12.5,0) {Team $T$};
   
   \draw [decorate,decoration={brace,amplitude=5pt,mirror,raise=4pt},yshift=0pt]
  (7,-6.5) -- (16.2,-6.5) node [black,midway,yshift=-1.6em] {$\kripke_{\varphi}$};
   
  \end{tikzpicture}